\documentclass[english]{llncs}
\usepackage[T1]{fontenc}
\usepackage[latin9]{inputenc}
\usepackage{color}
\usepackage{verbatim}
\usepackage{amsmath}
\usepackage{amssymb}
\usepackage{url}

\makeatletter
\usepackage[textwidth=20mm]{todonotes}

\makeatother

\usepackage{babel}
\usetikzlibrary{patterns}
\renewcommand{\epsilon}{\varepsilon}
\newcommand{\restateclaim}[3]{\noindent\textbf{#1 #2.}\, #3}
\newenvironment{restate}[3]{\restateclaim{#1}{#2}{#3} \vspace{1mm}\begin{proof}}%
{\qed\end{proof}\vspace{1mm}}

\global\long\def\P{P}
\global\long\def\PP{\mathcal{P}}
\global\long\def\slp{\mathrm{sLP}}
\global\long\def\ufpc{\mathrm{UFPcover}}
\global\long\def\T{\mathcal{T}}
\global\long\def\ct{c_{\mathrm{thres}}}
\global\long\def\reldate{{(1+\epsilon)}^{\left\lfloor \log_{1+\epsilon} \epsilon \cdot p_j/(1+\epsilon) \right\rfloor }}
\DeclareMathOperator{\OPT}{OPT}
\DeclareMathOperator{\opt}{opt}
\DeclareMathOperator{\size}{\sf size}
\DeclareMathOperator{\beg}{\sf begin}
\DeclareMathOperator{\rem}{\sf rem}
\DeclareMathOperator{\poly}{poly}
\DeclareMathOperator{\ex}{ex}

\begin{document}

\title{How Unsplittable-Flow-Covering helps 
\\Scheduling with Job-Dependent 
Cost Functions\thanks{%
Funded by the Go8-DAAD joint research cooperation scheme.}}
\author{Wiebke H\"ohn\inst{1}
\and Juli\'an Mestre\inst{2} \and Andreas Wiese\inst{3}}
\institute{Technische Universit{\"a}t Berlin, Germany.  \email{hoehn@math.tu-berlin.de}
\and The University of Sydney, Australia.  \email{mestre@it.usyd.edu.au}
\and Max-Planck-Institut f\"ur Informatik, Saarb\"ucken, Germany. \email{awiese@mpi-inf.mpg.de}}
\maketitle

%

\begin{abstract}
Generalizing many well-known and natural scheduling problems, scheduling with
job-specific cost functions has gained a lot of attention recently. In this
setting, each job incurs a cost depending on its completion time, given by a
private cost function, and one seeks to schedule the jobs to minimize the total
sum of these costs. The framework captures many important scheduling objectives
such as weighted flow time or weighted tardiness. Still, the general case as
well as the mentioned special cases are far from being very well understood
yet, even for only one machine. Aiming for better general understanding of this
problem,  in this paper we focus on the case of  uniform job release dates on
one machine for which the state of the art is a 4-approximation algorithm. This
is true even for a special case that is equivalent to the covering version of
the well-studied and prominent unsplittable flow on a path problem, which is
interesting in its own right. For that covering problem, we present a quasi-polynomial
time $(1+\epsilon)$-approximation algorithm that yields an
$(e+\epsilon)$-approximation for the above scheduling problem. 
Moreover, for the latter we devise the best possible resource augmentation
result regarding speed: a polynomial time algorithm which computes a solution
with \emph{optimal }cost at $1+\epsilon$ speedup. Finally, we present an
elegant QPTAS for the special case where the cost functions of the jobs fall
into at most $\log n$ many classes. This algorithm allows the jobs even to have
up to $\log n$ many distinct release dates.
%
\end{abstract}

\section{Introduction}

In scheduling, a natural way to evaluate the quality of a computed
solution is to assign a cost to each job which depends on its completion
time. The goal is then to minimize the sum of these costs. The function
describing this dependence may be completely different for each job.
There are many
well-studied and important scheduling objectives which can be cast
in this framework. Some of them are already very well understood,
for instance weighted sum of completion times $\sum_{j}w_{j}C_{j}$
for which there are polynomial time approximation schemes (PTASs)~\cite{PTAS-scheduling},
even for multiple machines and very general machine models. On the
other hand, for natural and important objectives such as weighted flow
time or weighted tardiness, not even a constant factor polynomial time
approximation algorithm is known, even on a single machine. In a recent
break-through result, Bansal and Pruhs presented a $O(\log\log P)$-approximation
algorithm~\cite{bansal2010geometry,bansal2012weighted}
for the single machine case where every job has its private cost function.
Formally, they study the General Scheduling Problem (GSP) where the
input consists of a set of jobs $J$ where each job $j\in J$ is specified
by a processing time $p_{j}$, a release date~$r_{j}$, and a
non-decreasing
cost function $f_{j}$, and the goal is to compute a preemptive schedule
on one machine which  minimizes $\sum_{j}f_{j}(C_{j})$
where $C_{j}$ denotes the completion time of job~$j$ in the computed
schedule. Interestingly, even though this problem is very general,
subsuming all the objectives listed above, the best known complexity
result for it is only strong $\mathsf{NP}$-hardness, so there might
even be a polynomial time $(1+\epsilon)$-approximation. 

Aiming to better understand GSP, in this paper we investigate the special case
that all jobs are released at time 0. This case is still strongly
$\mathsf{NP}$-hard~\cite{lawler77} and the currently best know approximation
algorithm for it is a $(4+\epsilon)$-approximation
algorithm~\cite{cheung2011primal,MestreVerschae2013}%
\footnote{In \cite{cheung2011primal} a prima-dual
$(2+\epsilon)$-approximation algorithm was claimed for this problem.
However, there is a error in the argumentation: there
are instances~\cite{MestreVerschae2013} where the algorithm constructs
a dual solution whose value differs from the optimal integral solution by
a factor of 4. %
}. As observed by Bansal and Verschae~\cite{BansalVerschaeObervation}, this problem is a generalization of the covering-version
of the well-studied Unsplittable Flow on a Path problem
(UFP)~\cite{amzingUFP2014,BCES2006,SODA-unsplit-flow,BSW11,CCGK2002,CMS07}.
The input of this problem consists of a path, each edge~$e$ having a
demand~$u_{e}$, and a set of tasks~$T$. Each task~$i$ is specified
by a start vertex $s_{i}$, an end vertex $t_{i}$, a size~$p_{i}$, and a
cost~$c_{i}$. In the covering version, the goal is to select a subset of
the tasks $T'\subseteq T$ which covers the demand profile, i.e., $\sum_{i\in T'\cap T_{e}}p_{i}\ge u_{e}$ where $T_{e}$ denotes all
tasks in~$T$ whose path uses~$e$. The objective is to minimize the
total cost $\sum_{i\in T'}c_{i}$.

This covering version of UFP has applications to resource allocation settings
such as workforce and energy management, making it an interesting problem in
its own right. For example, one can think of the tasks as representing time
intervals when employees are available, and one aims at providing certain service
level that changes over the day. UFP-cover is a generalization of the knapsack
cover problem~\cite{carr2000strengthening} and corresponds to instances of
GSP without release dates where the cost function of each job attains only the
values 0, some job-dependent value~$c_i$, and $\infty$.  The best known
approximation algorithm for UFP-cover is a
4-approximation~\cite{BBFNS2000journal,chakaravarthy2011resource}, which
essentially matches the best known result  for GSP without release dates.


\subsubsection{Our Contribution.}

In this paper we present several new approximation results for GSP without
release dates and some of its special cases. First, we give a
$(1+\epsilon)$-approximation algorithm for the covering version of UFP with
quasi-polynomial running time. 
Our algorithm follows the high-level idea of the known QPTAS for the
packing version~\cite{BCES2006}. Its key concept is to start with an edge in the middle and to
consider the tasks using it. One divides these tasks into groups, all tasks in
a group having roughly the same size and cost, and guesses for each group an
approximation of the capacity profile used by the tasks from that group. In the
packing version, one can show that by slightly underestimating the true profile
one still obtains almost the same profit as the optimum. For the covering
version, a natural adjustment would be to use an approximate profile which
\emph{over}estimates the true profile. However, when using only a polynomial
number of approximate profiles, it can happen that in the instance there are
simply not enough tasks from a group available so that one can cover the
overestimated profile which approximates the actual profile in the best
possible way.

We remedy this problem in a maybe counterintuitive fashion. Instead of
guessing an approximate upper bound of the true profile, we first
guess a \emph{lower} bound of it. Then we select tasks that cover
this lower bound, and finally add a small number of ``maximally long''
additional tasks.
Using this procedure, we cannot guarantee (instance-independently)
how much our selected tasks exceed the guessed profile on each edge.
However, we can guarantee that for the correctly guessed profile, we
cover at least as much as the optimum and pay only slightly more.
Together with the recursive framework from~\cite{BCES2006}, we obtain
a \mbox{QPTAS}. As an application, we use this algorithm to get a
quasi-polynomial time $(e  +\epsilon)$-approximation algorithm for GSP
with uniform release dates, improving the approximation ratio of the best
known polynomial time 4-approximation algorithm~\cite{cheung2011primal,MestreVerschae2013}.

Moreover, we consider a different way to relax the problem. Rather than
sacrificing a $1+\epsilon$ factor in the objective value, we present
a polynomial time algorithm that computes a solution with \emph{optimal}
cost but requiring a speedup of~$1+\epsilon$.
Such a result can be easily obtained for job-\emph{independent},
scalable cost functions using the PTAS in~\cite{MegowVerschae2013}
(a cost function $f$ is scalable if $f(c\,t)=\phi(c)\,f(t)$ for some suitable
function $\phi$ and all all $c,t\ge 0$).
In our case, however, the cost functions of the jobs can be much more
complicated and, even worse, they can be different for each job. Our
algorithm first imposes some simplification
on the solutions under consideration, at the cost of a $(1+\epsilon)$-speedup.
Then, we use a recently introduced technique to first guess 
a set of discrete intervals representing slots for large jobs and
then use a linear program to simultaneously assign large jobs into these
slots and small jobs into the remaining idle times~\cite{sviridenko2013approximating}.

An interesting open question is to design a (Q)PTAS for GSP without
release dates. As a first step towards this goal, recently Megow and
Verschae~\cite{MegowVerschae2013} presented a PTAS for minimizing
the objective function $\sum_{j}w_{j}g(C_{j})$ where each job~$j$ has
a private weight~$w_{j}$ but the function $g$ is identical for all jobs. In
Section~\ref{sec:few-classes} we present a QPTAS for a generalization
of this setting. Instead of only one function $g$ for all jobs, we allow up
to $(\log n)^{O(1)}$ such functions, each job using one of them, and we
even allow the jobs to have up to $(\log n)^{O(1)}$ distinct release dates.
Despite the fact that this setting is much more general, our algorithm is
very clean and easy to analyze. 


\subsubsection{Related Work.}

As mentioned above, Bansal and Pruhs present a $O(\log\log P)$-approximation
algorithm for GSP~\cite{bansal2010geometry}. Even for some well-studied
special cases, this is now the best known polynomial time approximation
result. For instance, for the important weighted flow time objective,
previously the best known approximation factors were $O(\log^{2}P)$,
$O(\log W)$ and $O(\log nP)$~\cite{BansalDhamdhere2007,ChekuriKhannaZhu2001}, where $P$ and $W$ denote
the ranges of the job processing times and weights, respectively.
A QPTAS with running time $n^{O_{\epsilon}(\log P\log W)}$ is also
known~\cite{chekuri2002approximation}. For the objective of minimizing the weighted
sum of completion times, PTASs are known, even for an arbitrary number of
identical and a constant number of unrelated machines~\cite{PTAS-scheduling}.

For the case of GSP with identical release dates, Bansal and
Pruhs~\cite{bansal2010geometry} give a 16-approximation algorithm. Later,
Shmoys and Cheung claimed a primal-dual $(2+\epsilon)$-approximation
algorithm~\cite{cheung2011primal}. However, an instance was later found where
the algorithm constructs a dual solution which differs from the best integral
solution by a factor 4~\cite{MestreVerschae2013}, suggesting that the
primal-dual analysis can show only an approximation ratio of 4. On the other
hand, Mestre and Verschae~\cite{MestreVerschae2013} showed that the local-ratio
interpretation of that algorithm (recall the close relation between the
primal-dual schema and the local-ratio technique~\cite{bar2005equivalence}) is
in fact a  pseudopolynomial time 4-approximation, yielding a
$(4+\epsilon)$-approximation in polynomial time.

As mentioned above, a special case of GSP with uniform release dates
is a generalization for the covering version of Unsplittable Flow on a Path.
For this special case, a 4-approximation algorithm is
known~\cite{BBFNS2000journal,chakaravarthy2011resource}. The packing version
is very well studied. After a series of papers on the problem and its special
cases~\cite{SODA-unsplit-flow,BSW11,CCGK2002,CMS07}, the currently best
known approximation results are a \mbox{QPTAS}~\cite{BCES2006} and a
$(2+\epsilon)$-approximation in polynomial time~\cite{amzingUFP2014}.

\section{Quasi-PTAS for UFP-Cover}
\label{sec:qptas-ufp}

In this section, we present a quasi-polynomial time $(1+\epsilon)$-approximation
algorithm for the UFP-cover problem. Subsequently, we show how it
can be used to obtain an approximation algorithm with approximation
ratio $e + \epsilon \approx2.718 + \epsilon$ and quasi-polynomial running
time for GSP without release dates. Throughout this section, we
assume that the sizes of the tasks are quasi-polynomially bounded.
Our algorithm follows the structure from the QPTAS for the packing
version of Unsplittable Flow on a Path
due to Bansal et al.~\cite{BCES2006}. First, we describe a recursive
exact algorithm with exponential running time. Subsequently, we describe
how to turn this routine into an algorithm with only quasi-polynomial
running time and an approximation ratio of $1+\epsilon$. 

For computing the exact solution (in exponential time) one can use the
following recursive algorithm: Given the path $G=(V,E)$, denote by $e_{M}$ the
edge in the middle of $G$ and let $T_{M}$ denote the tasks that use $e_{M}$.
Our strategy is to ``guess'' which tasks in $T_{M}$ are contained in $\OPT$,
the (unknown) optimal solution. Note that once these tasks are chosen, the
remaining problem splits into the two independent subproblems given by the
edges on the left and on the right of $e_{M}$,  respectively, and the  tasks
whose paths are fully contained in them. Therefore, we enumerate all subsets of
$T'_{M}\subseteq T_{M}$, denote by $\T_{M}$ the resulting set of sets. For each
set $T'_{M}\in\T_{M}$ we recursively compute the optimal solution for the
subpaths $\{e_{1},...,e_{M-1}\}$ and $\{e_{M+1},...,e_{|E|}\}$, subject to the
tasks in $T'_{M}$ being already chosen and that no more tasks from $T_{M}$ are
allowed to be chosen. The leaf subproblems are given when the path in the
recursive call has only one edge. Since $|E|=O(n)$ this procedure has a
recursion depth of $O(\log n)$ which is helpful when aiming at quasi-polynomial
running time. However, since in each recursive step we try each set
$T'_{M}\in\T_{M}$, the running time is exponential (even in one single step of
the recursion). To remedy this issue, we will show that for any set $\T_{M}$
appearing in the recursive procedure there is a set $\bar{\T}_{M}$ which is of
small size and which approximates $\T_{M}$ well. More precisely, we can compute
$\bar{\T}_{M}$ in quasi-polynomial time (and it thus has only quasi-polynomial
size) and there is a set $T_{M}^{*}\in\bar{\T}_{M}$ such that
$c(T_{M}^{*})\le(1+\epsilon)\cdot c(T_{M}\cap \OPT)$ and $T_{M}^{*}$ dominates
$T_{M}\cap \OPT$. For any set of tasks $T'$ we write $c(T'):=\sum_{i\in T'}c_i$, and
for two sets of tasks
$T_{1},T_{2}$, we say that $T_{1}$
\emph{dominates $T_{2}$ }if $\sum_{i\in T_{1}\cap T_{e}}d_{i}\ge\sum_{i\in
T_{2}\cap T_{e}}d_{i}$ for each edge $e$.  We modify the above procedure such
that we do recurse on sets in $\bar{\T}_{M}$ instead of $\T_{M}$. Since
$\bar{\T}_{M}$ has quasi-polynomial size, $\bar{\T}_{M}$ contains the mentioned
set $T_{M}^{*}$, and the recursion depth is $O(\log n)$, the resulting
algorithm is a QPTAS. In the sequel, we describe the above algorithm in detail
and show in particular how to obtain the set $\bar{\T}_{M}$.

\subsection{Formal Description of the Algorithm}

We use a binary search procedure to guess the optimal objective
value~$B$. First, we reject all tasks~$i$ whose cost is larger than~$B$
and select all tasks $i$ whose cost is at most $\epsilon B/n$. The latter
cost at most $n\cdot\epsilon B/n\le\epsilon B$ and thus only a factor
$1+\epsilon$ in the approximation ratio. We update the demand profile
accordingly. 

We define a recursive procedure $\ufpc(E',T')$ which gets as input
a subpath $E'\subseteq E$ of $G$ and a set of already chosen tasks~$T'$.
Denote by $\bar{T}$ the set of all tasks $i\in T\setminus T'$ such that
the path of $i$ uses only edges in $E'$. The output of $\ufpc(E',T')$ is
a $(1+\epsilon)$-approximation to the minimum cost solution for the
subproblem of selecting a set of tasks $T''\subseteq\bar{T}$
such that $T'\cup T''$ satisfy all demands of the edges in $E'$,
i.e., $\sum_{i\in (T'\cup T'')\cap T_e}p_{i}\ge d_{e}$ for each edge
$e\in E'$. Note that there might be no feasible solution for this
subproblem in which case we output~$\infty$. Let $e_{M}$ be the
edge in the middle of $E'$, i.e., at most $|E'|/2$ edges are on
the left and on the right of $e_{M}$, respectively. Denote by $T_{M}\subseteq\bar{T}$
all tasks in $\bar{T}$ whose path uses $e_{M}$. As described above,
the key is now to construct the set $\bar{\T}_{M}$ with the above
properties. Given this set, we compute $\ufpc(E_{L}',T'\cup T'_{M})$
and $\ufpc(E_{R}',T'\cup T'_{M})$ for each set $T'_{M}\in\bar{\T}_{M}$,
where $E'_{L}$ and $E'_{R}$ denote the subpaths of $E'$ on the
left and on the right of $e_{M}$, respectivley. We output 
\[
\min_{T'_{M}\in\bar{\T}_{M}}c(T'_{M})+\ufpc(E_{L}',T'\cup T'_{M})+\ufpc(E_{R}',T'\cup T'_{M}).
\]
For computing the set $\bar{\T}_{M}$, we first group the tasks in
$T_{M}$ into $(\log n)^{O(1)}$ many groups, all tasks in a group
having roughly the same costs and sizes. Formally, for each pair
$(k,\ell)$, denoting (approximately) cost $(1+\epsilon)^{k}$ and
size~$(1+\epsilon)^{\ell}$, we~define
\[
T_{(k,\ell)}:=\{i\in T_M:(1+\epsilon)^{k}\le c_{i}<(1+\epsilon)^{k+1}\wedge(1+\epsilon)^{\ell}\le p_{i}<(1+\epsilon)^{\ell+1}\}.
\]
Since the sizes of the tasks are quasi-polynomially bounded
and we preprocessed the weights of the tasks, we have $(\log n)^{O(1)}$
non-empty groups.

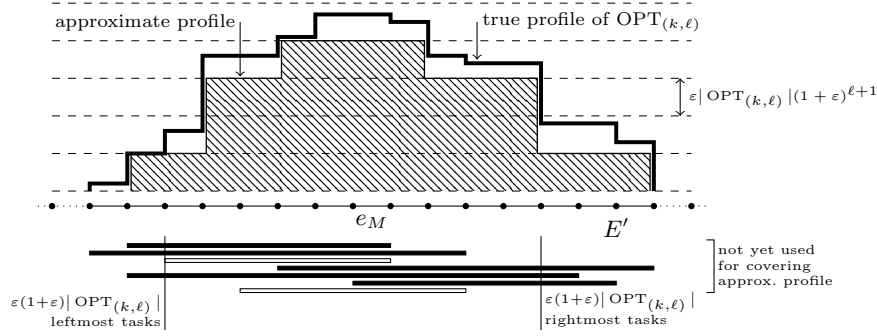
\begin{figure}[t]
\newcommand{\oneeps}{1.8}
\newcommand{\shift}{0.1}
\newcommand{\taskheight}{0.1}
\newcommand{\taskdiff}{0.2}
\hspace*{5mm}
\begin{tikzpicture}[scale=0.5,font=\scriptsize]
    \draw[line width=1.5pt]
        (0,1) -- ++(0,0.2) -- ++(1,0)
            -- ++(0,0.8)-- ++(1,0)
            -- ++(0,0.6) -- ++(1,0)
            -- ++(0,2) -- ++(2,0)
            -- ++(0,0.5) -- ++(1,0)
            -- ++(0,0.6) -- ++(2,0)
            -- ++(0,-0.2)
            -- ++(1,0) -- ++(0,-0.9)
            -- ++(1,0) -- ++(0,-0.2)
            -- ++(2,0) -- ++(0,-1.6)
            -- ++(2,0) -- ++(0,-0.5)
            -- ++(1,0) -- ++(0,-1.3);
    \draw[line width =0.5pt,pattern=north west lines]
            (1+\shift,1) -- ++(0,1)
            -- ++(2,0) -- ++(0,2)
            -- ++(2,0) -- ++(0,1)
            -- ++(4-2*\shift,0) -- ++(0,-1)
            -- ++(3,0) -- ++(0,-2)
            -- ++(3,0) -- ++(0,-1);
    \begin{scope}[yshift=-1.7cm]
        \draw (2,4*\taskdiff) rectangle ++(6,\taskheight)
                (4,0*\taskdiff) rectangle ++(6,\taskheight);
        \filldraw (0,5*\taskdiff) rectangle ++(10,\taskheight)
                (1,2*\taskdiff) rectangle ++(12,\taskheight)
                (5,3*\taskdiff) rectangle ++(10,\taskheight)
                (7,1*\taskdiff) rectangle ++(7,\taskheight)
                (1,6*\taskdiff) rectangle ++(7,\taskheight);
        \draw[font=\tiny, outer sep=0pt, inner sep=1.5pt, text width=2cm]
                (2,-1)
                node[anchor=south east, align=right]{$\epsilon(1+\epsilon)\big|\OPT_{(k,\ell)}\big|$ leftmost tasks~}
                -- ++(0,7*\taskdiff+1.1) 
                (12,-1)
                node[anchor=south west]{$\epsilon(1+\epsilon)\big|\OPT_{(k,\ell)}\big|$ rightmost tasks}
                -- ++(0,7*\taskdiff+1.1)
                (16.4,0) -- ++(0.2,0) --
                node[anchor=west, text width=1.7cm] {not yet used for covering  approx.\ profile}
                ++(0,7*\taskdiff)-- ++(-0.2,0);
    \end{scope}
    \foreach \y in {1,...,6}%
        {\draw[dashed, line width =0.01pt] (-1,\y) -- ++(17,0);}
     \draw[<->] (15.7,3) --
        node[anchor=west, font=\tiny] {$\epsilon\big|\OPT_{(k,\ell)}\big|(1+\epsilon)^{\ell+1}$} ++(0,1);
    \newcommand{\pathy}{0.6}
    \draw[dotted] (0,\pathy) -- ++(-1.5,0)
                        (15,\pathy) -- ++(1.5,0);
    \foreach \x in {0,...,14}
        {\draw (\x,\pathy) -- ++(1,0);}
    \foreach \x in {-1,...,16}
        {\filldraw (\x,\pathy) circle (2pt);}        
    \path[anchor=north, font=\footnotesize]
        (7.5,\pathy) node{$e_M$}
        (14,\pathy - 0.1) node{$E'$};
     \draw[<-]  (10.3,4.5) -- ++(0,1)  node[anchor=west, outer sep=1pt, inner sep=1pt, text width=3.2cm]
         {true profile  of $\OPT_{(k,\ell)}$};
     \draw[<-]  (4,4.05) -- ++(0,1.45)  node[anchor=east, outer sep=1pt, inner sep=1pt] {approximate profile};
\end{tikzpicture}
\caption{Construction from Lemma~\ref{lem:approx-group-sets}.}
\label{fig:step-profile}
\end{figure}

For each group $T_{(k,\ell)}$, we compute a set $\bar{\T}_{(k,\ell)}$
containing at least one set which is not much more expensive than
$\OPT_{(k,\ell)}:=\OPT\cap T_{(k,\ell)}$ and which dominates $\OPT_{(k,\ell)}$.
To this end, observe that the sizes of the tasks in $\OPT_{(k,\ell)}$ cover a
certain profile (see Figure~\ref{fig:step-profile}). Initially, we guess the
number of tasks in~$\OPT_{(k,\ell)}$, and if $|\OPT_{(k,\ell)}|\le
\tfrac{1}{\epsilon^2}$  then we simply enumerate all subsets of $T_{(k,\ell)}$
with at most $\tfrac{1}{\epsilon^2}$ tasks. Otherwise, we consider a polynomial
number of profiles that are potential approximations of the true profile
covered by~$\OPT_{(k,\ell)}$. To this end, we subdivide the (implicitly)
guessed height of the true profile evenly into~$\tfrac{1}{\epsilon}$ steps of
uniform height, and we allow the approximate profiles to use only those heights
while being monotonously increasing and decreasing before and  after $e_M$,
respectively (observe that also $\OPT_{(k,\ell)}$ has this property since all
its tasks use $e_M$). This leads to at most~$n^{O(1/\epsilon)}$ different
approximate profiles in total.

For each approximate profile we compute a set of tasks covering it
using LP-rounding.
The path of any task in $T_{(k,\ell)}$ contains the edge~$e_M$, and hence,
a task covering an edge~$e$ always covers all edges inbetween~$e$
and~$e_M$ as well.  Thus, when formulating the problem as an LP, it
suffices to introduce one constraint for the leftmost and one constraint for
the rightmost edge of each height in the approximated  profile. We compute
an extreme point solution of the LP and round up each of the at most
$\tfrac{2}{\epsilon}$ fractional variables. Since  $|\OPT_{(k,\ell)}|\ge \tfrac{1}{\epsilon^2}$ 
this increases the cost at most a factor $1+O(\epsilon)$ compared to the cost of the LP.


It is clear that the LP has a solution if the approximate profile is dominated
by the true profile. Among such approximate profiles, consider the one that
is closest to the latter. On each edge it would be sufficient to 
add $O(\epsilon\cdot\,\big|\OPT_{(k,\ell)}\big|)$ tasks from $T_{(k,\ell)}$ in
order to  close the remaining gap. This is due to our choice of the step size
of the approximate profile and the fact that all tasks in $T_{(k,\ell)}$ have
roughly the same size.  To this end, from the not yet selected tasks
in $T_{(k,\ell)}$  we add the $O(\epsilon\cdot |\OPT_{(k,\ell)}\big|)$ tasks with
the leftmost start vertex and the  $O(\epsilon\cdot |\OPT_{(k,\ell)}\big|)$ tasks
with the rightmost end vertex (see Figure~\ref{fig:step-profile}).
This costs again at most an $O(\epsilon)$-fraction of the cost so far. As a result,
on each edge $e$ we have either selected $O(\epsilon\cdot\,\big|\OPT_{(k,\ell)}\big|)$
additional tasks using it, thus closing the remaining gap, or
we have selected \emph{all} tasks from $T_{(k,\ell)}$ using $e$. In either case, 
the selected tasks dominate the tasks in $\OPT_{(k,\ell)}$, i.e., the true profile.
The above procedure is described in detail in Appendix~\ref{apx:QPTAS}. 
\newcommand{\lemapproxgroupsets}{%
Given a group $T_{(k,\ell)}$. There
is a polynomial time algorithm which computes a set of task sets \textup{$\bar{\T}_{(k,\ell)}$
}which contains a set $T_{(k,\ell)}^{*}\in\bar{\T}_{(k,\ell)}$ such
that $c(T_{(k,\ell)}^{*})\le(1+\epsilon)\cdot c(\OPT_{(k,\ell)})$
and \textup{$T_{(k,\ell)}^{*}$ }dominates $\OPT_{(k,\ell)}$.
}
\begin{lemma}
\label{lem:approx-group-sets}
\lemapproxgroupsets
\end{lemma}
We define the set $\bar{\T}_{M}$ by taking all combinations of selecting
exactly one set from the set $\bar{\T}_{(k,\ell)}$ of each group $T_{(k,\ell)}$.
Since there are $(\log n)^{O(1)}$ groups, by Lemma~\ref{lem:approx-group-sets}
the set $\bar{\T}_{M}$ has only quasi-polynomial size and it contains
one set $T_{M}^{*}$ which is a 
a good approximation to $T_{M}\cap \OPT$, i.e., the set $T_{M}^{*}$
dominates $T_{M}\cap \OPT$ and it is at most by a factor
$1+O(\epsilon)$ more expensive. Now each node in the recursion tree
has at most $n^{(\log n)^{O(1)}}$ children and, 
as argued above, the recursion depth
is $O(\log n)$. Thus, a call to $\ufpc(E,\emptyset)$ has quasi-polynomial
running time and yields a $(1+O(\epsilon))$-approximation for the overall problem.
\begin{theorem} \label{thm:qptas-ufp-cover}
For any $\epsilon>0$ there is a quasi-polynomial $(1+\epsilon)$-approximation algorithm for
UFP-cover if the sizes of the tasks are in a quasi-polynomial range.
\end{theorem}


Bansal~and~Pruhs~\cite{bansal2010geometry} give a $4$-approximation-preserving
reduction from GSP with uniform release dates to UFP-cover using geometric
rounding. Here we observe that if instead we use \emph{randomized geometric
rounding} \cite{smartcow},
then one can obtain an $e$-approximation-preserving reduction.
Together with our QPTAS for
UFP-cover, we get the following result, whose
proof we defer to Appendix~\ref{apx:QPTAS}.

\newcommand{\thmquasieapprox}{%
For any $\epsilon>0$ there is a quasi-polynomial time $(e+\epsilon)$-approximation algorithm for
GSP with uniform release dates.
}
\begin{theorem}
\label{thm:quasi-e-approx}
\thmquasieapprox
\end{theorem}

\section{General Cost Functions under Speedup}
\label{sec:general-cost-speedup}

We present a polynomial time algorithm which computes a solution for
an instance of GSP with uniform release dates whose cost is optimal and
which is feasible if the machine runs with speed $1+\epsilon$ (rather than
unit speed). 

Let $1 > \epsilon>0$ be a constant and assume for simplicity that
$\tfrac{1}{\epsilon}\in\mathbb{N}$. For our algorithm, we first prove some properties
that we can assume ``at $1+\epsilon$ speedup''; by this, we mean that there is
a schedule whose cost is at most the optimal cost (without enforcing these
restricting properties) and which is feasible if we increase the speed of the
machine by a factor $1+\epsilon$. Many statements are similar to properties
that are used in~\cite{PTAS-scheduling} for constructing PTASs for the problem
of minimizing the weighted sum of completion times.

For a given schedule denote by~$S_{j}$ and~$C_{j}$ the start and end times
of job~$j$ in a given schedule  (recall that we consider only non-preemptive
schedules). We define $C_{j}^{(1+\epsilon)}$ to be the smallest power
of $1+\epsilon$ which is not smaller than $C_{j}$, i.e.,
$C_{j}^{(1+\epsilon)}:=(1+\epsilon)^{\left\lceil \log_{1+\epsilon}C_{j}\right\rceil }$,
and adjust the objective function as given in the next lemma. Also,
we impose that jobs that are relatively large are not processed too
early; formally, they do not run before $\reldate$ which is the largest
power of  $1+\epsilon$ which is at most $\epsilon/(1+\epsilon)\cdot p_{j}$
(the speedup will compensate for the delay of the start time).

\newcommand{\lemsimplerobjective}{%
At $1+ O(\epsilon)$ speedup 
we  can use the objective
function $\sum_{j}f_{j}\big(C_{j}^{(1+\epsilon)}\big)$, instead of $\sum_{j}f_{j}(C_{j})$, and  assume $S_{j}\ge\reldate$ for each job $j$.
}
\begin{lemma}
\label{lem:simpler-objective}
\lemsimplerobjective 
\end{lemma}
Next, we discretize the time axis into intervals of the form $I_{t}:=[R_{t}, R_{t+1})$
where $R_{t}:=(1+\epsilon)^{t}$ for any integer~$t$. Note that $|I_{t}|=\epsilon\cdot R_{t}$.
Following Lemma~\ref{lem:simpler-objective}, to simplify the problem we
want to assign an artificial release date to each job~$j$. For each job~$j$,
we define $r(j):=\reldate$. Lemma~\ref{lem:simpler-objective} implies then
that we can assume $S_{j}\ge r(j)$ for each job~$j$. Therefore, we interpret
the value~$r(j)$ as the release date of job~$j$ and from now on disallow to
start job~$j$ before time~$r(j)$.

In a given schedule, we call a job \emph{$j$ large} if
$S_{j}\le\frac{1}{\epsilon^{3}}\cdot p_{j}$ and \emph{small }otherwise.
%
%
%
%
For the large jobs, we do not allow arbitrary starting times but we discretize
the time axis such that each interval contains only a constant number of
starting times for large jobs (for constant $\epsilon$). For the small jobs, we
do not want them to overlap over interval boundaries and we want that all
small jobs scheduled in an interval $I_{t}$ are scheduled during one
(connected) subinterval $I^{s}_{t}\subseteq I_{t}$.
\newcommand{\lemtechnicalsimplifications}{%
At $1+O(\epsilon)$ speedup we can assume that
\vspace*{-2mm}
\begin{itemize}
\item each small job starting during an interval~$I_{t}$ finishes during~$I_{t}$,
\item each interval~$I_{t}$ contains only~$O(\tfrac{1}{\epsilon^{3}})$ potential
start points for large jobs, 
and
\item for each interval $I_{t}$ there is a time interval $I^{s}_{t}\subseteq I_{t}$,
ranging from one potential start point for large jobs to another, which contains all
small jobs scheduled in~$I_{t}$ and no large jobs.
\end{itemize}
}
\begin{lemma}
\label{lem:technical-simplifications}
\lemtechnicalsimplifications
\end{lemma}
For the moment, let us assume that the processing times of the instance
are polynomially bounded. We will give a generalization to arbitrary
instances later.

Our strategy is the following: Since the processing times are bounded, the whole
schedule finishes within $\log_{1+\epsilon}(\sum_{j}p_{j})\le O_{\epsilon}(\log n)$
intervals. Ideally, we would like to guess the placement of all large jobs in the
schedule and then use a linear program to fill in the remaining small jobs. However,
this would result in $n^{O_{\epsilon}(\log n)}$ possibilities for the large jobs, which
is quasi-polynomial but not polynomial. Instead, we only guess the
\emph{pattern} of large-job usage for each interval. A pattern~$P$ for an interval
is a set of $O(\tfrac{1}{\epsilon^{3}})$
integers which defines the start and end times
of the large jobs which are executed during~$I_{t}$. Note that such a job might
start before~$I_{t}$ and/or end after~$I_{t}$. 
\begin{proposition}
\label{prop:patterns}For each interval $I_{t}$ there are only $N\in O_{\epsilon}(1)$
many possible patterns. The value~$N$ is independent of~$t$. 
\end{proposition}
We first guess all patterns for all intervals in parallel. Since there are
only $O_{\epsilon}(\log n)$ intervals, this yields
only $N^{O_{\epsilon}(\log n)}\in n^{O_{\epsilon}(1)}$
possible combinations for all patterns for all intervals. Suppose now that we
guessed the pattern corresponding to the optimal solution correctly.
Next, we solve a linear program that in parallel 
assigns large jobs to the slots specified by the pattern, and also, it assigns
small jobs into the remaining idle times on the intervals. Formally, we solve
the following LP. We denote by~$Q$ the set of all slots for large jobs, $\size(s)$
denotes the length of a slot~$s$, $\beg(s)$  its start time, and~$t(s)$ denotes the
index of the interval~$I_t$ that contains~$s$. For each interval~$I_{t}$
denote by $\rem(t)$ the remaining idle time for small jobs,
and consider these idle times as slots for small jobs, which we refer to by
their interval indices~$I:=\{1,\dots, \log_{1+\epsilon}(\sum_{j}p_{j})\}$.
For each pair of slot $s\in Q$  and job~$j\in J$, we introduce a variable~$x_{s,j}$
corresponding to assigning~$j$ to~$s$. Analogously, we use variables~$y_{t,j}$
for the slots in~$I$.
{\small
\begin{align}
\min~ \sum_{j\in J} & \bigg(\sum_{s\in Q} &&\hspace*{-9.5mm} f_j(R_{t(s)+1})\cdot x_{s,j}
    + \sum_{t\in I }  f_j(R_{t+1})\cdot y_{t,j}\bigg)
\label{eq:obj-function}\\
\sum_{s\in Q}x_{s,j} + \sum_{t\in I} y_{t,j}~ &=~1 &  & \forall j\in J
\label{eq:all-jobs-assigned}\\
\sum_{j\in J}x_{s,j}~ & \le~1 &  & \forall\, s\in Q
\label{eq:slot-constraint}\\
\sum_{j\in J}~p_{j}\cdot y_{t,j}~ & \leq~ \rem(t) &  & \forall \,t\in I
\label{eq:idle-constraint}\\[-2mm]
x_{s,j}\, & =\,0 &  & \forall\, s\in Q,\,\forall j\in J:~  r(j)> \beg(s)~\vee~p_{j}> \size(s)
\label{eq:x_sj}\\
y_{t,j}\, & =\,0 &  & \forall\, t\in I, \,\forall j\in J:~ r(j)> R_{t}~\vee~ p_{j}>\epsilon\cdot |I_{t}|.
\label{eq:y_tj}\\
x_{s,j},\, y_{t,j}\, &\geq\,0 &  & \forall\, s\in Q,\, \forall\, t\in I, \,\forall j\in J
\label{eq:nonneg}
\end{align}
\vspace*{-3mm}}

Denote the above LP by $\slp$. It has polynomial size and thus we
can solve it efficiently. Borrowing ideas from~\cite{shmoys1993approximation}
we round it to a solution that is not more costly and which can be
made feasible using additional speedup of $1+\epsilon$.
\newcommand{\lemslp}[1][]{%
Given a fractional solution $(x,y)$ to $\slp$. In polynomial time, we can compute
a non-negative integral solution $(x',y')$ whose cost is not larger than the cost
of $(x,y)$ and which fulfills the constraints (\ref{eq:all-jobs-assigned}),
(\ref{eq:slot-constraint}), (\ref{eq:x_sj}), (\ref{eq:y_tj}), \eqref{eq:nonneg} and
\begin{align}
\tag{\ref{eq:idle-constraint}a}
 &  & \sum_{j\in J}p_{j}\cdot y_{t,j} & \leq \rem(t)+\epsilon\cdot|I_{t}| &  & \forall\, t\in I.
 \ifthenelse{\equal{#1}{}}{\label{eq:idle-constraint-1}}{}
\end{align}
}
\begin{lemma}
\label{lem:slp}
\lemslp
\end{lemma}
In particular, the cost of the computed solution is no more than the
cost of the integral optimum and it is feasible under $1+O(\epsilon)$
speedup (accumulating all the speedups from the previous lemmas).
We remark that the technique of guessing patterns and filling them
in by a linear program was first used in~\cite{sviridenko2013approximating}.

For the general case, i.e., for arbitrary processing times, we first
show that at $1+\epsilon$ speedup, we can assume that for each job
$j$ there are only $O(\log n)$ intervals between $r(j)$ (the artificial
release date of $j$) and $C_{j}$. Then we devise a dynamic program
which moves from left to right on the time axis and considers sets
of $O(\log n)$ intervals at a time, using the above technique. See
Appendix~\ref{sec:appendix-thm-gsp-speedup}
for details.
\newcommand{\thmgspspeedup}{%
Let $\epsilon>0$. There is a polynomial time algorithm for GSP with
uniform release dates which computes a solution with optimal cost
and which is feasible if the machine runs with speed $1+\epsilon$.
}
\begin{theorem}
\label{thm:gsp-speedup}
\thmgspspeedup
\end{theorem}

\section{Few Classes of Cost Functions}\label{sec:few-classes}

In this section, we study the following special case of GSP with release dates. 
We assume that
each cost function $f_{j}$ can be expressed as $f_{j}=w_{j}\cdot g_{u(j)}$
for a job-dependent weight $w_{j}$, $k$ global functions $g_{1},...,g_{k}$,
and an assignment $u:J\rightarrow[k]$ of cost functions to jobs.
We present a QPTAS for this problem, assuming that $k=(\log n)^{O(1)}$
and that the jobs have at most $(\log n)^{O(1)}$ distinct release dates.
We assume that the job weights are in a quasi-polynomial range, i.e.,
we assume that there is an upper bound $W=2^{(\log n)^{O(1)}}$ for
the (integral) job weights. 

In our algorithm, we first round the values of the functions $g_{i}$
so that they attain only few values, $(\log n)^{O(1)}$ many. Then
we guess the $(\log n)^{O(1)}/\epsilon$ most expensive jobs and their costs.
For the remaining problem, we use a linear program. Since we rounded
the functions $g_{i}$, our LP is sparse, and by rounding an extreme
point solution we increase the cost by at most an $\epsilon$-fraction
of the cost of the previously guessed jobs, which yields an $(1+\epsilon)$-approximation
overall. 

Formally, we use a binary search framework to estimate the optimal
value~$B$. Having this estimate, we adjust the functions $g_{i}$
such that each of them is a step function with at most $(\log n)^{O(1)}$
steps, all being powers of $1+\epsilon$ or $0$. 
\newcommand{\lempolylogvalues}{%
At $1+\epsilon$ loss we can assume that for each $i\in[k]$ and each~$t$
it holds that $g_{i}(t)$ is either~$0$ or a power of $1+\epsilon$
in $\big[\frac{\epsilon}{n}\cdot\frac{B}{W},B\big)$.
}
\begin{lemma}
\label{lem:polylog-values} 
\lempolylogvalues
\end{lemma}
Our problem is in fact equivalent to assigning a due date $d_{j}$
to each job (cf.~\cite{bansal2010geometry}) such that the due dates
are \emph{feasible,} meaning that there is a preemptive schedule where every
job finishes no later than its due date, and the objective being $\sum_{j}f_{j}(d_{j})$.
The following lemma characterizes when a set of due dates is feasible.
{}
\begin{lemma}
[\cite{bansal2010geometry}]\label{lem:release-date-extension}
Given a set of jobs and a set of due dates. The due dates are feasible
if and only if for every interval $I=[r_{j},d_{j'}]$ for any two
jobs $j,j'$, the jobs in $X(I):=\{j:r_{j}\in I\}$ that are assigned
a deadline after $I$ have a total size of at least $\ex(I):=\max(\sum_{j\in X(I)}p_{j}-|I|,0)$.
That is, 
$\sum_{\bar{j}\in X(I):d_{\bar{j}}>d_{j'}}p_{\bar{j}}$ is at least~$\ex(I)$
for all intervals $I=[r_{j},d_{j'}]$.
\end{lemma}
Denote by $D$ all points in time where at least one cost function
$g_{i}$ increases. It suffices to consider only those values as possible
due dates. 
\begin{proposition}
\label{prop:all-due-dates-in-D}
There is an optimal due date assignment such that $d_{j}\in D$ for each job $j$. 
\end{proposition}
Denote by $R$ the set of all release dates of the jobs. Recall that
$|R|\le(\log n)^{O(1)}$. We guess now the $|D|\cdot|R|/\epsilon$
most expensive jobs of the optimal solution and their respective costs.
Due to the rounding in Lemma~\ref{lem:polylog-values} we have that
$|D|\le k\cdot\log_{1+\epsilon}(W\cdot n/\epsilon)=(\log n)^{O(1)}$
and thus there are only $O(n^{|D|\cdot|R|/\epsilon})=n^{(\log n)^{O(1)}/\epsilon}$
many guesses. 

Suppose we guess this information correctly. Let~$J_{E}$ denote the
guessed jobs and for each job $j\in J_{E}$ denote by $d_{j}$ the
latest time where it attains the guessed cost, i.e., its \emph{due
date}. Denote by~$\ct$ the minimum cost of a job in~$J_{E}$, according
to the guessed costs. The remaining problem consists in assigning
a due date $d_{j}\in D$ to each job $J\setminus J_{E}$ such that
none of these jobs costs more than $\ct$, all due dates together
are feasible, and the overall cost is minimized. We express this as
a linear program. In that LP, we have a variable $x_{j,t}$ for each 
pair of a job $j\in  J\setminus J_{E}$ and a due date $t\in D$
such that $j$ does not cost more than $\ct$ when finishing at time $t$.
We add the constraint $\sum_{t\in D}x_{j,t}=1$ for each job $j$,
modeling that the job has a due date, 
and one constraint for each interval $[r,t]$ with $r\in R$ and $t\in D$
to model the condition given by Lemma~\ref{lem:release-date-extension}.
See Appendix~\ref{apx:few-classes} for the full LP.
%

In polynomial time, we compute an extreme point solution $x^{*}$
for the~LP. It has at most $|D|\cdot|R|+|J\setminus J_{E}|$
many non-zeros. Each job $j$ needs at least one non-zero variable
$x_{j,t}^{*}$, due to the constraint $\sum_{t\in D}x_{j,t}=1$.
Thus, there are at most $|D|\cdot|R|$ fractionally assigned jobs,
i.e., jobs $j$ having a variable $x_{j,t}^{*}$ with $0<x_{j,t}^{*}<1$.
We define an integral solution by rounding 
$x^{*}$ as follows: For each job $j$ we set $d_{j}$ to be the
maximum value $t$ such that $x_{j,t}^{*}>0$. We round up at most
$|D|\cdot|R|$ jobs and after the rounding, each of them costs at
most $\ct$. Hence, those jobs cost at most an $\epsilon$-fraction
of the cost of guessed jobs ($J_{E}$).
\newcommand{\lemextremepointcost}{%
Denote by $c(x^{*})$ the cost of the solution $x^{*}$. We have that\\
\centerline{%
$\sum_{j\in J\setminus J_{E}}f_{j}(d_{j})~\le~ c(x^{*})+\epsilon\cdot\sum_{j\in J_{E}}
f_{j}(d_{j})$.
}
}
\begin{lemma}
\label{lem:extreme-point-cost}
\lemextremepointcost
\end{lemma}
Since $c(x^{*}) + \sum_{J_{E}}f_{j}(d_{j})$ is a
lower bound on the optimum, we obtain a $(1+\epsilon)$-approximation.
As there are quasi-polynomially many guesses for the expensive jobs
and the remainder can be done in polynomial time, we obtain a QPTAS.
\begin{theorem}
There is a QPTAS for GSP, assuming that each cost function $f_{j}$
can be expressed as $f_{j}=w_{j}\cdot g_{u(j)}$ for some job-dependent
weight $w_{j}$ and at most $k=(\log n)^{O(1)}$ global functions
$g_{1},...,g_{k}$, and that the jobs have at most $(\log n)^{O(1)}$
distinct release dates. 
\end{theorem}
%

\vspace{-2ex}
\bibliographystyle{plain}
\bibliography{../citations}

\begin{thebibliography}{10}

\bibitem{PTAS-scheduling}
F.~Afrati, E.~Bampis, C.~Chekuri, D.~Karger, C.~Kenyon, S.~Khanna, I.~Milis,
  M.~Queyranne, M.~Skutella, C.~Stein, and M.~Sviridenko.
\newblock Approximation schemes for minimizing average weighted completion time
  with release dates.
\newblock In {\em Proceedings of FOCS 1999}, pages 32--44, 1999.

\bibitem{amzingUFP2014}
A.~Anagnostopoulos, F.~Grandoni, S.~Leonardi, and A.~Wiese.
\newblock A mazing 2+$\epsilon$ approximation for unsplittable flow on a path.
\newblock In {\em Proceedings of SODA 2014}, 2014.

\bibitem{BCES2006}
N.~Bansal, A.~Chakrabarti, A.~Epstein, and B.~Schieber.
\newblock A quasi-{PTAS} for unsplittable flow on line graphs.
\newblock In {\em Proceedings of STOC 2006}, pages 721--729, 2006.

\bibitem{BansalDhamdhere2007}
N.~Bansal and K.~Dhamdhere.
\newblock Minimizing weighted flow time.
\newblock {\em ACM T.\ Alg.}, 3(4), 2007.

\bibitem{SODA-unsplit-flow}
N.~Bansal, Z.~Friggstad, R.~Khandekar, and R.~Salavatipour.
\newblock A logarithmic approximation for unsplittable flow on line graphs.
\newblock In {\em Proceedings of SODA 2009}, pages 702--709, 2009.

\bibitem{bansal2012weighted}
N.~Bansal and K.~Pruhs.
\newblock Weighted geometric set multi-cover via quasi-uniform sampling.
\newblock In {\em Proceedings of ESA 2012}, pages 145--156.

\bibitem{bansal2010geometry}
N.~Bansal and K.~Pruhs.
\newblock The geometry of scheduling.
\newblock In {\em Proceedings of FOCS 2010}, pages 407--414, 2010.
\newblock See also \url{http://www.win.tue.nl/~nikhil/pubs/wflow-journ3.pdf}.

\bibitem{BansalVerschaeObervation}
N.~Bansal and J.~Verschae.
\newblock Personal communication.

\bibitem{BBFNS2000journal}
A.~Bar-Noy, R.~Bar-Yehuda, A.~Freund, J.~Naor, and B.~Schieber.
\newblock A unified approach to approximating resource allocation and
  scheduling.
\newblock {\em J. ACM}, 48(5):1069--1090, 2001.

\bibitem{bar2005equivalence}
R.~Bar-Yehuda and D.~Rawitz.
\newblock On the equivalence between the primal-dual schema and the local ratio
  technique.
\newblock {\em SIAM J.\ Discrete Math.}, 19(3):762--797, 2005.

\bibitem{BSW11}
P.~Bonsma, J.~Schulz, and A.~Wiese.
\newblock A constant factor approximation algorithm for unsplittable flow on
  paths.
\newblock In {\em Proceedings of FOCS 2011}, pages 47--56, 2011.

\bibitem{carr2000strengthening}
R.~D. Carr, L.~K. Fleischer, V.~J. Leung, and C.~A. Phillips.
\newblock Strengthening integrality gaps for capacitated network design and
  covering problems.
\newblock In {\em Proceedings of SODA 2000}, pages 106--115, 2000.

\bibitem{chakaravarthy2011resource}
V.~T. Chakaravarthy, A.~Kumar, S.~Roy, and Yogish Sabharwal.
\newblock Resource allocation for covering time varying demands.
\newblock In {\em Proceedings of ESA 2011}, volume 6942 of {\em LNCS}, pages
  543--554. 2011.

\bibitem{CCGK2002}
A.~Chakrabarti, C.~Chekuri, A.~Gupta, and A.~Kumar.
\newblock Approximation algorithms for the unsplittable flow problem.
\newblock In {\em Proceedings of APPROX 2002}, volume 2462 of {\em LNCS}, pages
  51--66, 2002.

\bibitem{chekuri2002approximation}
C.~Chekuri and S.~Khanna.
\newblock Approximation schemes for preemptive weighted flow time.
\newblock In {\em Proceedings of STOC 2002}, pages 297--305, 2002.

\bibitem{ChekuriKhannaZhu2001}
C.~Chekuri, S.~Khanna, and A.~Zhu.
\newblock Algorithms for minimizing weighted flow time.
\newblock In {\em Proceedings of STOC 2001}, pages 84--93, 2001.

\bibitem{CMS07}
C.~Chekuri, M.~Mydlarz, and F.~Shepherd.
\newblock Multicommodity demand flow in a tree and packing integer programs.
\newblock {\em ACM T.\ Alg.}, 3, 2007.

\bibitem{cheung2011primal}
M.~Cheung and D.~B Shmoys.
\newblock A primal-dual approximation algorithm for min-sum single-machine
  scheduling problems.
\newblock In {\em Proceedings of APPROX 2011}, volume 6845 of {\em LNCS}, pages
  135--146. 2011.

\bibitem{smartcow}
M.-Y. Kao, J.~H. Reif, and S.~R. Tate.
\newblock Searching in an unknown environment: An optimal randomized algorithm
  for the cow-path problem.
\newblock {\em Inform. Comput.}, 131(1):63--79, 1996.

\bibitem{lawler77}
E.~L. Lawler.
\newblock A ``pseudopolynomial'' algorithm for sequencing jobs to minimize
  total tardiness.
\newblock {\em Ann. Discrete Math.}, 1:331--342, 1977.

\bibitem{LP1986}
L.~Lov{\'a}sz and M.~Plummer.
\newblock {\em Matching Theory}, volume~29 of {\em Annals of Discrete
  Mathematics}.
\newblock North-Holland, Amsterdam, 1986.

\bibitem{MegowVerschae2013}
N.~Megow and J~Verschae.
\newblock Dual techniques for scheduling on a machine with varying speed.
\newblock In {\em Proceedings of ICALP 2013}, volume 7965 of {\em LNCS}, pages
  745--756. 2013.

\bibitem{MestreVerschae2013}
J.~Mestre and J.~Verschae.
\newblock A 4-approximation for scheduling on a single machine with general
  cost function.
\newblock \url{http://arxiv.org/abs/1403.0298}.

\bibitem{shmoys1993approximation}
D.~B. Shmoys and {\'E}.~Tardos.
\newblock An approximation algorithm for the generalized assignment problem.
\newblock {\em Math. Program.}, 62(1-3):461--474, 1993.

\bibitem{sviridenko2013approximating}
M.~Sviridenko and A.~Wiese.
\newblock Approximating the configuration-{LP} for minimizing weighted sum of
  completion times on unrelated machines.
\newblock In {\em Proceedings of IPCO 2013}, volume 7801 of {\em LNCS}, pages
  387--398. 2013.

\end{thebibliography}

\newpage

\begin{appendix}
\noindent
\textbf{\Large Appendix}

\newcommand{\finish}{\hspace*{4cm} \textcolor{red}{\Large\bf . . .}}

\section{Omitted proofs from Section~\ref{sec:qptas-ufp}}\label{apx:QPTAS}

\newcommand{\prof}{Q}
\newcommand{\profset}{\mathcal{Q}}

In order to prove Lemma~\ref{lem:approx-group-sets}, we formally introduce the
notion of a \emph{profile}. A profile~$\prof:E'\to \mathbb{R}_{\geq 0}$ assigns
a \emph{height}~$\prof(e)$ to each edge~$e\in E'$, and a profile~$\prof$
\emph{dominates} a profile~$\prof'$ if $\prof(e)\geq \prof'(e)$ holds for all
$e\in E'$. The profile~$\prof_T$ \emph{induced} by the  tasks~$T$ is defined by
the heights $\prof_T(e)~:=~\sum_{i\in T_e} p_i$, where~$T_e$ 
denotes all tasks in~$T$ whose path contains the edge~$e$. Finally, a set of
tasks~$T$ dominates a set of tasks~$T'$ if $\prof_T$ dominates $\prof_{T'}$.

\medskip

\begin{restate}{Lemma}{\ref{lem:approx-group-sets}}{\lemapproxgroupsets}
In the first step, we guess the number of tasks in $\OPT_{(k,\ell)}
:=T_{(k,\ell)}\cap \OPT$. Abusing notation, 
we write $\OPT_{(k,\ell)}$ also for the total cost of the tasks in $\OPT_{(k,\ell)}$.
If~$|\OPT_{(k,\ell)}|$ is smaller
than~$\tfrac{1}{\epsilon^2}$ then we can guess an optimal
set~$\OPT_{(k,\ell)}$. Otherwise, we will consider a polynomial number of
certain approximate profiles one of which underestimates the unknown true
profile induced by~$\OPT_{(k,\ell)}$ by at most
$O(\epsilon)\cdot\big|\OPT_{(k,\ell)}\big|$. For each approximate profile we will
compute a cover of cost  no more than $1+O(\epsilon)$ the optimum, and in case
of the profile being close to the true profile, we can extend this solution to
a cover of the true profile by adding only
$O(\epsilon)\cdot\big|\OPT_{(k,\ell)}\big|$ more tasks.

Several arguments in the remaining proof are based on the structure
of~$T_{(k,\ell)}$ and the resulting structure of the true profile
$\prof_{\OPT_{(k,\ell)}}$. Since all tasks in~$T_{(k,\ell)}$ containing the
edge~$e_M$ and spanning a subpath of~$E'$, the height of the profile
$\prof_{\OPT_{(k,\ell)}}$ is unimodular: It is non-decreasing until $e_M$ and
non-increasing after that; see Figure~\ref{fig:step-profile}. In particular, a
task that covers a certain edge~$e$ covers all edges in between~$e$ and~$e_M$ as
well.

For the approximate profiles, we restrict to heights from
\begin{align*}
\mathcal{H}
~:=~\left\{j\cdot\epsilon\cdot|\OPT_{(k,\ell)}|\cdot(1+\epsilon)^{\ell+1}\,\Big|~
        j\in\left\{0,1,\dots,\tfrac{1}{\epsilon}\right\} \right\}\,.
\end{align*}
Moreover, aiming to approximate the true profile, we only take into account
profiles in which the edges have non-decreasing and non-increasing height
before and after~$e_M$ on the path, respectively. Utilizing the natural ordering
of the edges on the path, we formally define the set~$\profset$ of approximate
profiles as follows
\begin{align*}
\profset~:=~\left\{ \prof ~\left|~
\begin{minipage}{7.6cm}
\small
$\prof(e) \in\mathcal{H} ~~\forall\, e\in E'
~~\wedge~~ \prof(e)\leq\prof(e') ~~\forall\, e<e'\leq e_M$\\
\hspace*{2.85cm}
$\wedge~~~ \prof(e)\geq\prof(e') ~~\forall\, e_M\leq e<e'$
\end{minipage}
\right.\right\}\,.
\end{align*}
Since $|\OPT_{(k,\ell)}| \cdot (1+\epsilon)^{\ell+1}$ is an upper bound on the maximum
height of  $\prof_{\OPT_{(k,\ell)}}$, there is a profile~$\prof^*\in\profset$ which
is dominated by~$\prof_{\OPT_{(k,\ell)}}$ and for which the
gap $\prof_{\OPT_{(k,\ell)}}(e) -\prof(e)$ does not
exceed $\epsilon\cdot|\OPT_{(k,\ell)}|\cdot(1+\epsilon)^{\ell+1}$ for all $e\in E'$.
Observe that by construction, an approximate profile can have 
at most~$|\mathcal{H}|$ edges at
which it jumps from one height to a larger one, and
analogously, it can have at most~$|\mathcal{H}|$ edges where it can
jump down to some smaller height. Hence, $\profset$ contains at
most $n^{2\,|\mathcal{H}|} = n^{2/\epsilon}$ profiles.

For each approximate profile $\prof\in\profset$, we compute a cover based
on LP rounding. To this end,
we denote by~$e_{L}(h)$ and~$e_{R}(h)$ the first and last edge~$e\in E'$
for which $\prof(e)\geq h$, respectively.
Note that by the
structure of the paths of tasks in~$T_{(k,\ell)}$, in fact every set
of tasks covering~$e_{L}(h)$  also covers all edges between~$e_M$
and~$e_{L}(h)$ by at least the same amount, and analogously for~$e_{R}(h)$.
Regarding the LP-formulation, this allows us to only require a sufficient
covering of the edges~$e_{L}(h)$ and~$e_{R}(h)$ rather than of all edges. 
Denoting by~$P_i$ the path of a task~$i$, and by~$x_i$ the decision
variable representing its selection for the cover,  we formulate the LP as follows
\begin{align*}
\min \sum_{i \in T_{(k,\ell)}}& c_i\cdot x_i \\
\sum_{i \in T_{(k,\ell)}:e_{L}(h)\in P_i}  \hspace*{-5mm} x_i\cdot p_i ~& \geq~h &&  \forall\,h\in \mathcal{H}\\
\sum_{i \in T_{(k,\ell)}:e_{R}(h)\in P_i}  \hspace*{-5mm} x_i\cdot p_i ~& \geq~h &&  \forall\,h\in \mathcal{H}\\
0~\leq x_i ~& \leq~1 && \forall\, i \in T_{(k,\ell)}\,.
\end{align*}

If there exists a feasible solution to the LP, we round up all fractional
values~$x_{i}^{*}$ (i.e., values $x_{i}^{*}\in(0,1)$) of some optimal
extreme point solution~$x^{*}$, and we choose the corresponding
tasks as a cover for~$\prof$ and denote them by~$T^{*}$. Since the
LP has only $2|\mathcal{H}|=\tfrac{2}{\epsilon}$ more constraints
than variables, its optimal extreme point solutions contain at most~$\tfrac{2}{\epsilon}$
fractional variables. Hence, the additional cost incurred by the rounding
does not exceed~$\tfrac{2}{\epsilon}(1+\epsilon)^{k+1}$, where the
latter term is the maximum task cost in~$T_{(k,\ell)}$. Let us assume
for calculating the cost of the computed solution that $Q=Q^{*}$. Then,
the cost of the selected tasks is at most 
\begin{align*}
\sum_{i\in T_{(k,\ell)}}c_{i}\cdot x_{i}^{*}\,+\,\tfrac{2}{\epsilon}(1+\epsilon)^{k+1}
& ~\leq~\OPT_{(k,\ell)}\,+\,2\epsilon\cdot\left|\OPT_{(k,\ell)}\right|\cdot(1+\epsilon)^{k+1}\\
& ~\leq~\big(1+2\epsilon(1+\epsilon)\big)\cdot\OPT_{(k,\ell)}\,,
\end{align*}
where the first and second inequality follows from $\big|\OPT_{(k,\ell)}\big|\geq\tfrac{1}{\epsilon^{2}}$
and from the minimum task weight in~$T_{(k,\ell)}$, respectively,
and moreover, the first inequality uses that~$\prof=Q^{*}$ is dominated
by~$\prof_{\OPT_{(k,\ell)}}$.

After covering~$\prof$ in the first step with~$T^{*}$, in the
second step, we extend this cover by additional edges~$A^{*}\subseteq T_{(k,\ell)}\setminus T^{*}$.
We define the set~$A^{*}$ to be the $\epsilon\,(1+\epsilon)\cdot\big|\OPT_{(k,\ell)}\big|$
tasks in $T_{(k,\ell)}\setminus T^{*}$ with the leftmost start vertices
and the $\epsilon\,(1+\epsilon)\cdot\big|\OPT_{(k,\ell)}\big|$ tasks in
$T_{(k,\ell)}\setminus T^{*}$ with the rightmost end vertices. We add
$T^{*}\cup A^{*}$ to the set~$\bar{\T}_{(k,\ell)}$. 

Assume that $Q=Q^{*}$. Then the above LP has a feasible solution and in particular the 
We claim that the computed tasks $T^{*}\cup A^{*}$ dominate $\OPT_{(k,\ell)}$.
Firstly, observe that any set of $\epsilon\,(1+\epsilon)\cdot\big|\OPT_{(k,\ell)}\big|$
tasks from~$T_{(k,\ell)}$ has a total size of at least the gap between
two height steps from~$\mathcal{H}$. Hence, if an edge~$e$ is
covered by that many edges from~$A^{*}$ and $Q=Q^{*}$ then we know
that $Q_{T^{*}\cup A^{*}}(e)\ge\prof_{\OPT_{(k,\ell)}}(e)$. 

On the other hand, if an edge~$e$ is covered by less than $\epsilon\,(1+\epsilon)\cdot\big|\OPT_{(k,\ell)}\big|$
tasks from~$A^{*}$, we know that there exists no further task in
$T_{(k,\ell)}\setminus(T^{*}\cup A^{*})$ whose path contains~$e$.
Otherwise, 
this would be a contradiction to the choice of the tasks $A^{*}$
being the $\epsilon\,(1+\epsilon)\cdot\big|\OPT_{(k,\ell)}\big|$ ones with the
leftmost start and rightmost end vertices,  respectively.
Thus, since in this second case $T^{*}\cup A^{*}$ contains all tasks
that cover~$e$, we have that $Q_{T^{*}\cup A^{*}}(e)\ge\prof_{\OPT_{(k,\ell)}}(e)$. 

Finally, the total cost of~$A^{*}$ does not exceed
\begin{align*}
2\epsilon\,(1+\epsilon)\cdot\big|\OPT_{(k,\ell)}\big|\cdot(1+\epsilon)^{k+1}
~\leq~2\epsilon\,(1+\epsilon)^{2}\cdot\OPT_{(k,\ell)}\,.
\end{align*}
and thus the total cost of $T^{*}\cup A^{*}$ is upper-bounded by
$$\big(1+2\epsilon(1+\epsilon)(2+\epsilon)\big)\cdot \OPT_{(k,\ell)}\,.$$
We complete the proof by redefining~$\epsilon$ appropriately. 
\end{restate}

\begin{restate}{Theorem}{\ref{thm:quasi-e-approx}}{\thmquasieapprox}
        The heart of the proof is an $e$-approximation-preserving reduction from GSP
        with uniform release dates to UFP-cover. Although here we
        develop a randomized algorithm, we note that the reduction can be
        de-randomized using standard techniques.

        Given an instance of the scheduling problem we construct an instance of
        UFP-cover as follows. For ease of presentation, we take our path $G=(V,E)$ to
        have vertices $0, 1, \ldots, P$; towards the end, we explain how to obtain an
        equivalent and more succinct instance. For each $i = 1, \ldots, P$, edge $e
        =(i-1,i)$ has demand $u_e = P - i$.

        The reduction has two parameters, $\gamma > 1$  and $\alpha \in [0,1]$,
        which will be chosen later to minimize the approximation guarantee. For
        each job $j$, we define a sequence of times $t_0^j, t_1^j, t_2^j,
        \ldots, t^k_j$ starting from $0$ and ending with $P+1$ such that the
        cost of finishing a job in between two consecutive times differs by at
        most a factor of $\gamma$. Formally, $t_0^j = 0$, $t_k^j = P+1$ and
        $t_i^j$ is the first time step such that $f(t_i^j) >
        \gamma^{i - 1 +\alpha}$. For each $i> 0$ such that $t_{i-1}^j < t_i^j$, we
        create a task covering the interval $[t_{i-1}^j, t_i^j - 1]$ having demand
        $p_j$ and costing $f_j(t_i^j - 1)$.

        Given a feasible solution of the UFP-cover instance, we claim that we
        can construct a feasible schedule of no greater cost. For each job $j$,
        we consider the right-most task chosen (we need to pick at least one
        task from each job to be feasible) in the UFP-cover solution and assign
        to $j$ a due date equal to the right endpoint of the task. Notice that
        the cost of finishing the jobs by their due date equals the total cost
        of these right-most tasks. By the feasibility of the UFP-cover
        solution, it must be the case that for each time $t$, the total
        processing volume of jobs with a due date of $t$ or great is at least
        $T-t+1$. Therefore, scheduling the jobs according to earliest due date
        first, yields a schedule that meets all the due date. Therefore, the
        cost of the schedule is at most the cost of the UFP-cover instance.

        Conversely, given a feasible schedule, we claim that, if $\alpha$ is
        chosen uniformly at random and set $\gamma = e$, then there is a solution of
        the UFP-cover instance whose expected cost is at most $e$ times more
        expensive that the cost of the schedule. For each job $j$, we pick all
        the tasks whose left endpoint is less than or equal to the completion
        time of $j$. It follows that the UFP-cover solution is feasible. Let
        $f_j(C_j)$ be the cost incurred by $j$. For a fixed $\alpha$, let the
        most expensive task induced by $j$ cost $f_j(C_j)
        \gamma^{\beta}$. Notice that $\beta$ is also uniformly distributed in
        $[0,1]$. The combined expected cost of all the tasks induced by $j$ is
        therefore
        \begin{equation*}
                \int_0^1 f_j(C_j) \left( \gamma^\beta + \gamma^{\beta-1} + \cdots \right) d \beta = f_j(C_j) \frac{\gamma}{\ln \gamma},
        \end{equation*}
        which is minimum at $\gamma = e$. By linearity of expectation, we get that
        the total cost of the UFP-cover solution is at most an $e$ factor larger than
        the cost of the schedule.

        To de-randomize the reduction, and at the expense of adding another
        $\epsilon'$ to the approximation factor, one can discretize the random
        variable $\alpha$, solve several instances, and return the one producing the
        best solution. Finally, we mention that it is not necessary to construct the
        full path from $0$ to $P$. It is enough to keep the vertices where tasks
        start or end. Stretches where no task begins or end can be summarized by an
        edge having demand equal to the largest demand in that stretch.

        Applying the $e$-approximation-preserving reduction and then running
        the $(1+ \epsilon)$-approximation of Theorem~\ref{sec:qptas-ufp} finishes the proof.
\end{restate}

\section{Omitted proofs from Section~\ref{sec:general-cost-speedup}}

%
In the following lemmas, we show different properties that we can assume
at a speedup of~$1+\epsilon$. In fact, each property requires to increase
the speed by another factor of $1+\epsilon$. Compared to the initial unit
speed, the final speed will be some power of $1+\epsilon$. Technically,
we consolidate the resulting polynomial in~$\epsilon$ to some~$\epsilon' = O(\epsilon)$,
achieving all properties of the lemmas at speed~$1+\epsilon'$.
\medskip

\begin{restate}{Lemma}{\ref{lem:simpler-objective}}{\lemsimplerobjective}
Consider some job~$j$ with completion time~$C_j$ in an arbitrary schedule
at unit speed. At speed~$1+\epsilon$,
time~$C_{j}^{(1+\epsilon)}$ corresponds to
\begin{align*}
(1+\epsilon)^{\left\lceil \log_{1+\epsilon}\tfrac{C_j}{1+\epsilon} \right\rceil}
~=~(1+\epsilon)^{\left\lceil \log_{1+\epsilon} C_j\right\rceil -1}
~\leq~C_j\,,
\end{align*}
and hence, the ensued cost never exceeds the original cost.

Regarding the second point of the lemma, we observe that running
a job~$j$ of processing time~$p_j$ at speed~$1+\epsilon$ allows
for an additional idle time of length ${\epsilon/(1+\epsilon)\cdot p_j}$
compared to running it at unit speed. Hence, in case that $S_j < \reldate$
we can set its start time to~$\reldate$ without exceeding its unit
speed completion time.

Hence, we can make the assumptions of the lemma at a total
speedup of $(1+\epsilon)^2$, which is $1+O(\epsilon)$ under our assumption that $\epsilon <1$, so the lemma follows.
\end{restate}

\noindent
Lemma~\ref{lem:technical-simplifications} is restated in a slightly stronger 
way, the statement given here immediately implies the version in the main
part of the paper.
\bigskip

\begin{restate}{Lemma}{\ref{lem:technical-simplifications}}{%
Let~$I_{t,k} :=[R_{t,k}, R_{t,k+1})$ where
$R_{t,k}:=(1+k\cdot\tfrac{1}{4}\,\tfrac{\epsilon^{4}}{1+\epsilon})\,R_{t}$
for $t\in\mathbb{N}$ and $k\in\{0,...,4\,\frac{1+\epsilon}{\epsilon^{3}}\}$.
\begin{itemize}
\item At $1+\epsilon$ speedup any small job starting during an interval~$I_{t}$ finishes in~$I_{t}$.
\item At $1+\epsilon$ speedup we
can assume that each large job starts at some point in time~$R_{t,k}$ and every
interval~$I_{t,k}$ is used by either only small jobs or by one large job
or it is empty.
\item For each interval $I_{t}$ there is a time interval $I_{t,k,\ell}:=[R_{t,k},R_{t,\ell})$
with $0 \leq k\leq \ell \leq 4\,\frac{1+\epsilon}{\epsilon^{3}}$
during which no large jobs are scheduled, and no small jobs are scheduled
during $I_{t}\setminus I_{t,k,\ell}$.
\end{itemize}
}
%
%
Consider a small job that is started in~$I_t$ and that is completed in
some later interval. By definition, its length is at most~$\epsilon^{3}\cdot R_{t+1}$.
At speed~$1+\epsilon$, the interval~$I_t$ provides an additional idle time 
of length
\begin{align*}
\left(1-\tfrac{1}{1+\epsilon}\right)\epsilon \cdot R_t
~=~ \tfrac{\epsilon^2}{1+\epsilon}\cdot R_t\,,
\end{align*}
and the length of the small job reduces to at most~$\epsilon^{3}\cdot R_{t}$.
Since for sufficiently small~$\epsilon$ it holds that
$\tfrac{\epsilon^2}{1+\epsilon} \geq \epsilon^{3}$, the small job can be
scheduled during the idle time, and hence, it finishes in~$I_t$. 
%

Regarding the second point of the lemma, we observe that the length
of a large job starting during~$I_t$ is at least~$\epsilon^3\cdot R_t$ by
definition. When running a large job at speed~$1+\epsilon$, its processing
time reduces by at least~$\epsilon^4/(1+\epsilon)\cdot R_t$ which equals four
times the gap between two values~$R_{t,k}$. If~$I_{t,k}$ and~$I_{t,\ell}$
are the first and last interval in~$I_t$ used by some large job~$j$ in a
unit speed schedule then, at speed~$1+\epsilon$, we can start~$j$
at time $R_{t,k+2}$, and it will finish no later than~$R_{t,\ell-1}$
or it will finish in some later interval~$I_s$, $s>t$.
In case of job~$j$ finishing in~$I_t$, the speedup allows us to assume~$j$ 
to block the interval~$[R_{t,k+2}, R_{t,\ell-1})$, and we know that no other job
is scheduled in this interval.

Otherwise, if~$j$ finishes in some later~$I_s$,
let~$I_{s,m}$ be the subinterval of its completion. Since~$j$ is not necessarily
large in~$I_s$, its reduce in runtime due a speedup may only be marginal
with respect to~$I_s$. In~$I_{s,m}$, the job~$j$ may be followed by a set of
$s$-small jobs and a $s$-large job (both possibly not existing). Analogously
to the above argumentation, at a speedup of~$1+\epsilon$, we can start the
$s$-large job at time~$R_{s,m+2}$, and the interval~$I_{s,m+1}$ becomes
empty. We use this interval to schedule the small jobs from~$I_{s,m}$. This
delays their start, however, they still finish in~$I_s$ which is sufficient:  By
the first part of Lemma~\ref{lem:simpler-objective}, we can calculate the
objective function as if every job  finished at the next larger value~$R_{r}$
after its actual completion time, i.e., at the end of the interval~$I_{r}$ during
which it finishes. Hence, within an interval~$I_{r}$ we can rearrange the
intervals~$I_{r,k}$ without changing the cost.
This completes the proof of the second part of the lemma.

The proof of the third part is a straight-forward implication
of its second part. By this we can assume that all small jobs are contained
in intervals~$I_{t,k}$ that contain no large jobs. Applying again the first part
of Lemma~\ref{lem:simpler-objective}, we can rearrange those intervals  in
such a way that they appear consecutively.
\end{restate}

\begin{restate}{Lemma}{\ref{lem:slp}}{\lemslp[nolabel]}
The proof follows the general idea of~\cite{shmoys1993approximation}.
Given some fractional solution~$(x,y)$ to the
sLP \eqref{eq:all-jobs-assigned}\,--\,\eqref{eq:nonneg},
%
we construct
a fractional matching~$M$ in a bipartite
graph~$G=(V\cup W, E)$. For each job~$j\in J$ and for each large slot~$s\in Q$,
we introduce vertices~$v_j\in V$ and~$w_s\in W$, respectively. Moreover, for
each slot of small jobs~$t\in I$, we add~$k_t:=\big\lceil\sum_{j\in J}y_{t,j}\big\rceil$
vertices $w_{t,1},\dots,w_{t,k_t}\in W$. We introduce an edge~$(v_j, w_s)\in E$
with cost~$f_j(R_{t(s)+1})$ for all job-slot pairs for which $x_{s,j}>0$, and we
choose it to an extent of~$x_{s,j}$ for~$M$. Regarding the vertices
$w_{t,1},\dots,w_{t,k_t}$, we add edges in the following way. We first sort all
jobs~$j$ with $y_{t,j}>0$ in non-increasing order of their length~$p_j$, and we
assign them greedily to $w_{t,1},\dots,w_{t,k_t}$; that is, we choose the first
vertex~$w_{t,\ell}$ which has not yet been assigned one unit of fractional jobs,
we assign as much as possible of~$y_{t,j}$ to it, and if necessary, we assign the
remaining part to the next vertex~$w_{t,\ell+1}$. Analogously to the above edges,
we define the cost of an edge~$(v_j, w_{t,\ell})$ to be~$f_j(R_{t+1})$, and we
add it fractionally to~$M$ according to the fraction~$y_{t, \ell,j}$ of~$y_{t,j}$ the job was
assigned to~$w_{t,\ell}$ by the greedy assignment.
Note that  $p_{t,\ell}^{\min}\geq p_{t,\ell+1}^{\max}$ for $\ell=1,\dots,k_t-1$
where~$p_{t,\ell}^{\min}$ and~$p_{t,\ell}^{\max}$ are the minimum and
maximum length of all jobs (fractionally) assigned to~$w_{t,\ell}$, respectively.

By construction,~$M$ is in fact a fractional matching, i.e.,  for every
vertex~$v_j\in V$ the set~$M$ contains edges whose chosen fractions add
up to exactly~$1$. Moreover, the total cost of~$M$ equals the cost of the
solution~$(x,y)$. Due to standard matching theory, we know that there also
exists an integral matching~$M'$ in~$G$ whose cost does not exceed the
cost of~$M$, and since~$G$ is bipartite, we can compute such a
matching in polynomial time, see e.g., \cite{LP1986}.
We translate~$M$ back into an integral solution~$(x',y')$ of the LP
where we set~$y_{t,j}=1$ for every edge~$(v_j, w_{t,\ell})$ in~$M$.
It remains to show that~$(x',y')$ satisfies~(\ref{eq:all-jobs-assigned}),
(\ref{eq:slot-constraint}), \eqref{eq:idle-constraint-1}, (\ref{eq:x_sj}), (\ref{eq:y_tj})
and~(\ref{eq:nonneg}). All constraints but~\eqref{eq:idle-constraint-1}
are immediately satisfied by construction. In order to show
that~\eqref{eq:idle-constraint-1} is satisfied observe that
\begin{align*}
\sum_{j\in J} p_j \cdot y'_{t,j}
~~ &\leq~~ \sum_{\ell=1}^{k_t} p_{t,\ell}^{\max} 
~~\leq~~  p_{t,1}^{\max}  + \sum_{\ell=2}^{k_t} p_{t,\ell}^{\max} 
~~\leq~~  \epsilon\cdot|I_t|  +\sum_{\ell=1}^{k_t-1} p_{t,\ell}^{\min}\\ 
&\leq~~  \epsilon\cdot|I_t|  +\sum_{\ell=1}^{k_t-1}\hspace*{-4mm}\sum_{\genfrac{}{}{0pt}{2}{j\in J:}{~~~(v_j, w_{t,\ell})\in E}}
    \hspace*{-5mm} p_j\cdot y_{t,\ell,j}
~~\leq~~  \epsilon\cdot|I_t|  +\sum_{\ell=1}^{k_t}\hspace*{-4mm}\sum_{\genfrac{}{}{0pt}{2}{j\in J:}{~~~(v_j, w_{t,\ell})\in E}}
    \hspace*{-5mm} p_j\cdot y_{t,\ell,j}\\
&=~~  \epsilon\cdot|I_t|  +\sum_{j\in J}  p_j\cdot y_{t,j}
~~\leq~~  \epsilon\cdot|I_t|  + \rem(t)\,,
\end{align*}
where the third inequality follows from~(\ref{eq:y_tj}).
\end{restate}

\section{Proof of Theorem~\ref{thm:gsp-speedup} for general processing times}
\label{sec:appendix-thm-gsp-speedup}

In this section, we provide the missing technical details which
allow to generalize the proof of Theorem~\ref{thm:gsp-speedup}
from polynomially bounded processing times to general processing
times.

\medskip
\restateclaim{Theorem}{\ref{thm:gsp-speedup}}{\thmgspspeedup}
\medskip

\noindent
We first prove that at $1+\epsilon$ speedup, we can assume that jobs
``live'' for at most $O(\log n)$ intervals, i.e., for each job $j$ there
are only $O(\log n)$ intervals between $r(j)$ (the artificial release
date) and $C_{j}$. Then, we devise a dynamic program which moves
on the time axis from left to right, considers blocks of $O(\log n)$
consecutive intervals at once and computes a schedule for them using
the approach from Section~\ref{sec:general-cost-speedup}.

\begin{lemma}
\label{lem:halflife-of-jobs}
At $1+\epsilon$ speedup we can assume that $\frac{C_{j}}{r(j)}\le q(n):=\tfrac{1}{\epsilon^{3}}\,n +(1+\epsilon)^{5}$.
Thus, $[r(j),C_{j})$ is contained in
at most $K\le O_{\epsilon}(\log n)$ intervals.
\end{lemma}
\begin{proof}
By using $1+\epsilon$ speedup we create an idle time of
$$
|I_{t}| - \tfrac{1}{1+\epsilon}\cdot|I_{t}|
~=~\tfrac{\epsilon}{1+\epsilon}\cdot|I_{t}|
~=~\frac{\epsilon^{2}\,(1+\epsilon)^{t}}{1+\epsilon}
~=~\epsilon^{2}\,(1+\epsilon)^{t-1}
$$
in each interval~$I_{t}$.
Then, the idle time during the interval
$I_{t+s}$ with $s:=\log_{1+\epsilon}\big(\tfrac{n}{\epsilon^{3}}\big)+3$
can fit all jobs~$j$ with $r(j)\le R_{t}$:
$$
\epsilon^{2}\,(1+\epsilon)^{t+s-1}
~=~\epsilon^{2}\,(1+\epsilon)^{t+\log_{1+\varepsilon}(n/\epsilon^{3})+2}
~=~ n\cdot\tfrac{1}{\epsilon}\,(1+\epsilon)^{t+2}
~\geq\sum_{j:r(j)\le R_{t}} \!\! p_{j}\,,
$$
where the last inequality is a consequence of Lemma~\ref{lem:simpler-objective}
which implies in the case of ${r(j)\le R_{t}}$
$$
t
~\geq~ \left\lfloor \log_{1+\epsilon}\left(\tfrac{\epsilon}{1+\epsilon}\cdot p_{j}\right) \right\rfloor
~\geq~  \log_{1+\epsilon}\left(\tfrac{\epsilon}{1+\epsilon}\cdot p_{j}\right) -1
~=~  \log_{1+\epsilon}\left(\epsilon\cdot p_{j}\right) -2\,.
$$

Since all jobs~$i$ with~$r(i)\le R_{t-1}$ can be assumed to be scheduled
in the idle time of some earlier interval if necessary, we can assume
$R_{t-1}<r(j)\le R_{t}$, and hence,
$$
\frac{C_{j}}{r(j)}
~\leq~ \frac{R_{t+s+1}}{R_{t-1}}
~=~ (1+\epsilon)^{s+2}
~=~ \tfrac{1}{\epsilon^{3}}\cdot n +(1+\epsilon)^{5}\,.
$$
In particular, it is sufficient to consider~$s+2=O_{\epsilon}(\log n)$
intervals for processing a job.
\qed
\end{proof}
Throughout the remainder of this section we denote by
$K:=\big\lceil\log_{1+\epsilon}(q(n))\big\rceil\in O_{\epsilon}(\log n)$
where~$q(n)$ is the polynomial from Lemma~\ref{lem:halflife-of-jobs}.
Thus,~$K$ denotes the number of intervals between the time~$r(j)$
and the completion time~$C_{j}$ of each job~$j$.

If after the assumption of Lemma~\ref{lem:halflife-of-jobs} there
is a point in time~$s$ that will \emph{not }schedule any
job, i.e., there is no job~$j$ with  $s\in[r(j),r(j)\cdot q(n))$,
then we divide the instance into two independent pieces.
\begin{proposition}
Without loss of generality we can assume that the union of all intervals $\bigcup_{j}[r(j),r(j)\cdot q(n))$
is a (connected) interval.
\end{proposition}
For our dynamic program we subdivide the time axis into \emph{blocks}.
Each block~$B_{i}$ consists of the intervals $I_{i\cdot K},...,I_{(i+1)\cdot K-1}$.
The idea is that in each iteration the DP schedules the jobs released
during a block~$B_{i}$ in the intervals of block~$B_{i}$ and
block $B_{i+1}$. So in the end, the intervals of each block $B_{i+1}$
contain jobs released during~$B_{i}$ and $B_{i+1}$.

To separate the jobs from both blocks we prove the following lemma.
\begin{lemma}
\label{lem:idle-time-blocks}
At $1+\epsilon$ speedup we can assume
that during each interval $I_{t}$ in a block $B_{i+1}$ there are
two subintervals $[a_t,b_t),[b_t,c_t)\subseteq I_{t}$ such that 
\begin{itemize}
\item during $[a_t,b_t)$ only small jobs from block $B_{i}$ are scheduled
and during $I_{t}\setminus[a_t,b_t)$ no small jobs from block $B_{i}$
are scheduled,
\item during $[b_t,c_t)$ only small jobs from block $B_{i+1}$ are scheduled
and during $I_{t}\setminus[b_t,c_t)$ no small jobs from block $B_{i+1}$
are scheduled,
\item $a_t,b_t,c_t$ are of the form
$(1+z\cdot\tfrac{\epsilon^{4}}{4\,(1+\epsilon)^2})\cdot R_{t}$
for  $x\in\mathbb{N}$ and $z\in\{0,1,...,\frac{4\,(1+\epsilon)^2}{\epsilon^{3}}\}$
(so possibly $[a_t,b_t)=\emptyset$ or $[b_t,c_t)=\emptyset$).
\end{itemize}
\end{lemma}
\begin{proof}
Based on Lemma~\ref{lem:technical-simplifications} we can assume that
all small jobs that are started within~$I_t$ also finish in~$I_t$; moreover,
they are processed in some interval~$I_{t,k,\ell}\subseteq I_t$ which contains
no large jobs (see Lemma~\ref{lem:technical-simplifications} for the notation).
By Lemma~\ref{lem:halflife-of-jobs}, the interval~$I_t$ can be
assumed to contain only small jobs with release date in~$B_i$ and~$B_{i+1}$,
and by Lemma~\ref{lem:simpler-objective} we know that we can rearrange the
jobs in~$I_t$ without changing the cost. Hence, for proving the lemma it is
sufficient to show that we can split~$I_{t,k,\ell}$ at some of the discrete points
given in lemma, such that the small jobs released in~$B_i$ and~$B_{i+1}$ are
scheduled before and after this point, respectively.

The interval~$I_{t,k,\ell}$ starts
at~$(1+\tfrac{1}{4}\, k\cdot\epsilon^{4}/(1+\epsilon))\cdot R_{t}$ and its length
is some integral multiple of~$\tfrac{1}{4}\, \epsilon^{4}/(1+\epsilon)\cdot R_{t}$.
At a speedup of~$1+\epsilon$, the interval~$I_{t,k,\ell}$ provides additional idle
time of length at least~$\tfrac{1}{4}\, \epsilon^{4}/(1+\epsilon)^2\cdot R_{t}$
(if~$I_{t,k,\ell}$ is not empty), which equals the step width of the discrete
interval end points required in the lemma. Hence, by scheduling all small
jobs released in~$B_i$ and~$B_{i+1}$ at the very beginning and very end
of~$I_{t,k,\ell}$, there must be 
point in time $s:=(1+z\cdot\tfrac{\epsilon^{4}}{4\,(1+\epsilon)^2})\cdot R_{t}$
with~$z\in\{0,1,...,\frac{4\,(1+\epsilon)^2}{\epsilon^{3}}\}$ which lies in the idle
interval between the two groups of small jobs. Finally, if setting~$a_t$ and~$c_t$
to the start and end of~$I_{t,k,\ell}$, respectively, and if choosing $b_t:=s$,
we obtain intervals as claimed in the lemma.
\qed
\end{proof}

Using Lemma~\ref{lem:halflife-of-jobs} we devise a dynamic program.
We work again with patterns for the intervals. Here a pattern for
an interval~$I_{t}$ in a block $B_{i}$ denotes~$O(\epsilon)$
integers which define 
\begin{itemize}
\item the start and end times of the large jobs from $B_{i-1}$ which are
executed during~$I_{t}$,
\item the start and end times of the large jobs from $B_{i}$ which are
executed during~$I_{t}$,
\item $a_t,b_t,c_t$ according to Lemma~\ref{lem:idle-time-blocks}, implying slots for small jobs.
\end{itemize}
Denote by $\bar{N}$ the number of possible patterns for an interval
$I_{t}$ according to this definition. Similarly as in Proposition~\ref{prop:patterns}
we have that $\bar{N}\in O_{\epsilon}(1)$ and $\bar{N}$ is independent
of~$t$.

Each dynamic programming cell is characterized by a tuple $(B_{i},\P_{i})$
where~$B_{i}$ is a block during which at least one job is released or during the
block thereafter, 
and~$\P_i$ denotes a pattern for all intervals of block~$B_{i}$.
For a pattern~$\P_i$, we denote by~$Q_i(\P_i)$ and~$Q_{i-1}(\P_i)$
the set of slots in~$B_i$ which are reserved for large jobs released
in~$B_{i-1}$ and~$B_i$, respectively. Moreover, for some interval~$I_t$
in~$B_i$ let~$D_{i-1,t}(\P_i)$ and~$D_{i,t}(\P_{i})$ be the two slots
for small jobs from~$B_{i-1}$ and~$B_i$, respectively.
The number of DP-cells is polynomially bounded as there are
only~$n$ blocks during which at least one job is released and, as in
Section~\ref{sec:general-cost-speedup}, the number of patterns for
a block is bounded by $\bar{N}^{O_{\epsilon}(\log n)}\in n^{O_{\epsilon}(1)}$.

The subproblem encoded in a cell $(B_{i},\P_i)$ is to schedule all
jobs $j$ with $r(j)\ge I_{i\cdot K}$ during $[R_{i\cdot K},\infty)$ while
obeying the pattern~$\P_i$ for the intervals $I_{i\cdot K},...,I_{(i+1)\cdot K-1}$.
To solve this subproblem we first enumerate all possible patterns
$\P_{i+1}$ for all intervals of block~$B_{i+1}$. Suppose that we guessed
the pattern~$\P_{i+1}$ corresponding to the optimal solution of the subproblem
given by the cell $(B_{i},\P_{i})$. Like in Section~\ref{sec:general-cost-speedup}
we solve the problem of scheduling the jobs of block~$B_{i}$ according
to the patterns~$\P_i$ and $\P_{i+1}$ by solving and rounding a linear
program of the same type as $\slp$. Denote by $\opt(B_{i},\P_i,\P_{i+1})$
the optimal solution to this subproblem.
\begin{lemma}
\label{lem:DP-subproblem}Given a DP-cell $(B_{i},\P_i)$ and a pattern~$\P_{i+1}$.
There is a polynomial time algorithm which computes a solution to the problem
of scheduling all jobs released during $B_{i}$ according to the patterns~$\P_i,\P_{i+1}$
which
\begin{itemize}
\item does not cost more than $\opt(B_{i},\P_{i},\P_{i+1})$ and
\item is feasible if during $B_{i}$ and $B_{i+1}$ the speed of the machine
is increased by a factor $1+\epsilon$.
\end{itemize}
\end{lemma}
\begin{proof}
The proof works analogously to the proof of Lemma~\ref{lem:slp}. We
formulate the following LP for (fractionally) solving the problem
{\small
\begin{align}
&&&\hspace*{-5cm}
\min \sum_{j\in J_i}
    \bigg( \hspace*{-3mm} \sum_{\genfrac{}{}{0pt}{2}{s\in Q_{i}(P_i)}{\hspace*{7mm}\cup Q_i(P_{i+1})}}
        \hspace*{-6mm}   f_j(R_{t(s)+1})\cdot x_{s,j}
    ~+ \sum_{t=i\cdot K}^{(i+2)\cdot K-1} \hspace*{-3mm} f_j(R_{t+1})\cdot y_{t,j} \bigg)
\label{eq:obj-function-general}\\[1mm]
\hspace*{-5mm}
\sum_{\genfrac{}{}{0pt}{2}{s\in Q_{i}(P_i)}{\hspace*{7mm}\cup Q_i(P_{i+1})}}  \hspace*{-6mm}  x_{s,j}
    ~+ \sum_{t=i\cdot K}^{(i+2)\cdot K-1} \hspace*{-3mm}  y_{t,j}~ & =~1 && \forall\, j\in J_i
\label{eq:all-jobs-assigned-general}\\
\sum_{j\in J_i}x_{s,j}~ & \le~1 &  & \forall\, s\in Q_{i}(P_i)\cup Q_i(P_{i+1})
\label{eq:slot-constraint-general}\\
\sum_{j\in J_i}~p_{j}\cdot y_{t,j}~ & \leq~ |D_{i,t}(P_{i(t)})| && \forall \,t\in \{i\cdot K, \dots, (i+2)\cdot K -1\}
\label{eq:idle-constraint-general}\\
x_{s,j}\, & =\,0 &  & \forall\, j\in J_i,\,\forall\, s\in Q:~  r(j)> \beg(s)
\label{eq:x_sj-general}\\
&&& \hspace*{2.65cm}\vee~p_{j}> \size(s)  \nonumber\\
y_{t,j}\, & =\,0 &  & \forall\, t\in I, \,\forall j\in J_i:~ r(j)> R_{t}
\label{eq:y_tj-general}\\
&&& \hspace*{1.7cm}~\vee~ p_{j}>\epsilon\cdot |I_{t}| \nonumber\\
x_{s,j},\,y_{t,j}\, & \geq\,0 &  & \forall\, j\in J_i,~ \forall\, s\in Q_{i}(P_i)\cup Q_i(P_{i+1}),
\label{eq:nonneg-general}\\
&&& \forall \,t\in \{i\cdot K, \dots, (i+2)\cdot K -1\}\,.\nonumber
\end{align}
}%
where~$J_i\subseteq J$ denotes the set of all jobs~$j$ with $r(j)\in B_i$,
and~$i(t)$ is the index of the block the interval~$I_t$ is contained in.

This LP has exactly the same structure as $\slp$~\eqref{eq:obj-function}\,--\,\eqref{eq:nonneg}
and hence, we obtain an analogous result to Lemma~\ref{lem:slp}. This
means that given a fractional solution~$(x,y)$ to the above LP, we can
construct an integral solution~$(x',y')$ which is not more costly than~$(x,y)$,
and which fulfills all constraints~\eqref{eq:all-jobs-assigned-general}\,--\,\eqref{eq:nonneg-general}
with~\eqref{eq:idle-constraint-general} being replaced by the relaxed
constraint
\begin{align*}
\sum_{j\in J_i}~p_{j}\cdot y_{t,j}~ & \leq~ |D_{i,t}(P_{i(t)})| \,+\, \epsilon\cdot |I_t|
&& \forall \,t\in \{i\cdot K, \dots, (i+2)\cdot K -1\}\,.
\end{align*}
However, at speedup of~$1+\tfrac{\epsilon}{1-\epsilon}\in 1+O(\epsilon)$, an
interval~$I_t$ provides an additional idle time of $\epsilon\cdot|I_t|$ which allows
for scheduling the potential job volume of  by which we may exceed the capacity
of the interval. Due to Lemma~\ref{lem:simpler-objective}, this does not increase
the cost of the schedule which concludes the proof.
\qed
\end{proof}

%

By definition of the patterns, an optimal solution~$\OPT(B_{i+1},\P_{i+1})$
is independent of the patterns that have been chosen for earlier blocks. This is
simply due to the separately reserved slots for jobs from different blocks within
each pattern, i.e., a slot in~$B_{i+1}$ which is reserved for jobs from~$B_i$
cannot be used by jobs from~$B_{i+1}$ in any case. Hence, $\OPT(B_{i},\P_{i})$
decomposes into~$\OPT(B_{i+1},\P_{i+1})$ and $\opt(B_{i},\P_i,\P_{i+1})$
for a pattern~$\P_{i+1} \in \PP_{i+1}$ which leads to the lowest cost, where~$\PP_{i+1}$
denotes the set of all possible patterns for block~$B_{i+1}$. Thus, formally it holds
\begin{align}
\label{eq:dp-substructure}
\hspace*{-2mm} \OPT(B_{i},\P_i) ~=\min_{\P_{i+1}\in\PP_{i+1}} 
     \OPT(B_{i+1},\P_{i+1}) ~+~ \opt(B_{i},\P_i,\P_{i+1})\,.
\end{align}

This observation of an optimal substructure allows to easily formulate a~DP.
We interpret each cell~$(B_i, \P_i)$ as a node in a graph, and we add
an edge between cells~$(B_i, \P_i)$ and~$(B_{i+1}, \P_{i+1})$ for all $\P_{i}\in \PP_{i}$
and $\P_{i+1}\in \PP_{i+1}$. For each triple $B_i$, $\P_i$, $\P_{i+1}$
we compute a solution using Lemma~\ref{lem:DP-subproblem}, and we
assign the cost of this solution to the edge $\big((B_i, \P_i),(B_{i+1}, \P_{i+1})\big)$.
Due to~\eqref{eq:dp-substructure}, a minimum cost path in this $O(\poly(n))$
size graph corresponds to a scheduling solution whose cost, at speed~$1+\epsilon$,
does not exceed the optimal cost at unit speed. This implies Theorem~\ref{thm:gsp-speedup}.


\section{Omitted proofs and LP from Section~\ref{sec:few-classes}}\label{apx:few-classes}
\begin{restate}{Lemma}{\ref{lem:polylog-values}}{\lempolylogvalues}
Denote by~$g^{(1+\epsilon)}_i$ the rounded cost functions for~$i\in[k]$, i.e.,
formally we define
\begin{align*}
g^{(1+\epsilon)}_i(t) ~:=~ \begin{cases}
\min\left\{(1+\epsilon)^{\left\lceil \log_{1+\epsilon}\left(g_i(t)\right)\right\rceil},~B \right\}
& \text{, if } g_i(t) >   \tfrac{\epsilon}{n} \cdot\tfrac{B}{W}  \\
\tfrac{\epsilon}{n} \cdot\tfrac{B}{W} & \text{, if } 0< g_i(t)\leq  \tfrac{\epsilon}{n} \cdot\tfrac{B}{W} \\
0 & \text{, if } g_i(t)=0 \,.
\end{cases}
\end{align*}
Consider some optimal schedule with completion time~$C_j$ for~$j\in J$.
Then it holds that
\begin{align*}
\sum_{j\in J} w_j \cdot g^{(1+\epsilon)}_{u(j)}\left(C_j\right)
~&\leq\sum_{\genfrac{}{}{0pt}{2}{j\in J:~0<g_{u(j)}\left(C_j\right)}{~~~\leq(\epsilon\cdot B)/(n\cdot W)}}
        \hspace*{-6mm} w_j \cdot \tfrac{\epsilon\cdot B}{n\cdot W}
    ~+~(1+\epsilon)\cdot\hspace*{-6mm}\sum_{\genfrac{}{}{0pt}{2}{j\in J:~g_{u(j)}(C_j)}{>(\epsilon\cdot B)/(n\cdot W)}}
        \hspace*{-5mm} w_j \cdot g_{u(j)}\\
~&\leq~ \epsilon\cdot B\cdot \tfrac{1}{n}\hspace*{-5mm}
        \sum_{\genfrac{}{}{0pt}{2}{j\in J:~0<g_{u(j)}\left(C_j\right)}{~~~\leq(\epsilon\cdot B)/(n\cdot W)}}
        \hspace*{-6mm} \tfrac{w_j}{W}
    ~+~(1+\epsilon)\cdot B\\
~&\leq~ \epsilon\cdot B ~+~(1+\epsilon)\cdot B
~~=~~ (1+2\,\epsilon)\cdot B\,.
\end{align*}
The lemma follows by redefining~$\epsilon$.
\end{restate}

\medskip\noindent
At this point, we give the full formulation of the LP described in short in
Section~\ref{sec:few-classes}.  After guessing the $|D|\cdot |R|/\epsilon$
most expensive jobs~$J_{E}$,  the solution to this LP is the basis for scheduling
the remaining problem.
{\small
\begin{align}
&&&\hspace*{-4.5cm}
 \min\sum_{j\in J\setminus J_{E}}~\sum_{t\in D}~x_{j,t}\cdot f_{j}(t)
 \label{eq:obj-fct-few-cost}\\
 \sum_{\genfrac{}{}{0pt}{2}{j\in (J\setminus J_{E})}{ ~~~~\cap X([r,t])}}\sum_{\genfrac{}{}{0pt}{2}{t'\in D:}{t'>t}}\,
 p_{j}\cdot x_{j,t'}+\hspace*{-4mm}\sum_{\genfrac{}{}{0pt}{2}{j\in J_{E}\cap X([r,t]):}{d_{j}>t}}
 \hspace*{-4mm} p_{j}~ & \ge~ \ex([r,t]) ~~&  & \forall \,r\in R~~\forall\, t\in D
 \label{eq:ex-cover-few-cost}\\
 \sum_{t\in D}x_{j,t}~ & =~1 &  & \forall j\in J\setminus J_{E}
 \label{eq:all-jobs-assigned-Nicole}\\
 x_{j,t}~ &=~0 &  & \forall j\in J\setminus J_{E}~~\forall t\in D: 
 \label{eq:zeros-few-cost}\\
&&& r_{j}+p_{j}> t ~\vee~ w_{j}\, g_{u(j)}(t)>\ct\nonumber\\
 x_{j,t}~ & \ge~0 &  & \forall j\in J\setminus J_{E}~~\forall t\in D
\label{eq:nonneg-few-cost} 
\end{align}
}%
Denote by~$x^*$ an optimal solution to this LP.

\bigskip

\begin{restate}{Lemma}{\ref{lem:extreme-point-cost}}{\lemextremepointcost}
Informally, the proof of this lemma has already been given in
the main part of the paper. Here, we only add the missing formal
step. We define an integral solution by simply rounding up the
solution~$x^{*}$. Formally, for each job~$j$ we define~$d_{j}$
to be the maximum value~$t$ such that $x_{j,t}^{*}>0$. As observed
in Section~\ref{sec:few-classes}, the solution~$x^{*}$ has at
most~$|D|\cdot|R|$ fractional entries $0<x_{j,t}^{*}<1$, and
hence, the rounding affects at most~$|D|\cdot|R|$ variables
whose corresponding cost $x_{j,t}\cdot f_{j}(t)$ do not
exceed~$\ct$ after the rounding. Thus, in the resulting schedule
we have
\begin{align*}
\sum_{j\in J\setminus J_{E}} f_{j}(d_{j})~
& \leq~ c\big(x^{*}\big) + |D|\cdot|R|\cdot\ct
~~\leq~~ c\big(x^{*}\big) + \epsilon\cdot\tfrac{1}{\epsilon}\cdot|D|\cdot|R|\cdot\ct\\
& \leq~ c\big(x^{*}\big) + \epsilon\cdot \sum_{j\in J_{E}} f_{j}(d_{j}) \,.
\end{align*}
This completes the proof.
\end{restate}

\end{appendix}

\end{document}